%% file: main.tex
\documentclass[12pt]{article}
\usepackage{natbib}
\usepackage{amsmath,amssymb,amsthm,bm}
\usepackage{xcolor}
\usepackage{a4wide}
\usepackage{setspace}
\usepackage{subcaption}
\usepackage{tikz,url}
\usepackage{comment}
\usepackage{multirow}
\usepackage{hyperref}

\newtheorem{theorem}{Theorem}
\newtheorem{lemma}{Lemma}
\newtheorem{coro}{Corollary}

\newcommand{\chref}[2]{\href{#1}{\textcolor{blue}{#2}}}

\newtheorem{assumption}{Assumption}

\newcommand{\argmin}{\operatornamewithlimits{argmin}}

\newcount\blind

\newcommand{\defurl}[3]{%
  \expandafter\newcommand\csname #1\endcsname{%
    \ifnum\blind=1 #3\else #2\fi
  }%
}

\defurl{urlprereg}{https://osf.io/feu5j}{https://osf.io/feu5j/overview?view_only=97d687da56fe42d9a90f216931564225}
\defurl{urlproject}{https://osf.io/72abk}{https://osf.io/72abk/overview?view_only=496a34e9d5e24012b055ff07cf898f6d}

\title{Disentangling Causal Mechanisms in Conjoint Experiments Using Mediation\footnote{
  \ifnum\blind=1
    An anonymized version of our pre-registration plan can be found \chref{\urlprereg}{here}. An anonymized version of our survey instrument can be found \chref{\urlproject}{here} and is discussed in Appendix~\ref{app:survey_instrument}. The pre-registered experiment was approved by the appropriate IRB. Open-source software will be made available to implement this method.
  \else
    The pre-registered experiment in this study was approved by the University of Texas at Austin IRB (STUDY00007312). The pre-registration plan can be found \chref{\urlprereg}{here}. We thank Alex Coppock, Libby Jenke, Connor Jerzak, Tony Linero, Mats Stensrud,  Dustin Tingley, Matthew Tyler and participants at TexMeth 2026, MPSA 2026, and the Texas Junior Methods Workshop for comments on an earlier draft. Aklin gratefully acknowledges the support of the Swiss National Science Foundation (Grant \#10001834, \textit{Managing a Just Transition to Net-Zero} project).
  \fi
}}
\author{
\ifnum\blind=1
    ~
\else
    Michaël Aklin\footnote{College of Management, EPFL.}\hspace{1.1em} 
    Max Goplerud\footnote{Department of Government, University of Texas at Austin}\hspace{1.1em}
    Nicole E. Pashley\footnote{Department of Statistics, Rutgers University.}\hspace{1.1em}
    Jenna Salzman\footnote{Department of Government, University of Texas at Austin.}
\fi
}
\date{\today}

 \newcommand{\ind}{\perp\!\!\!\!\perp} 
\newcommand{\E}{\mathbb{E}}
\newcommand{\I}[1]{\bm{1}\{#1\}}
\newcommand{\bSi}{\bm{S}_i}
\newcommand{\bXi}{\bm{X}_i}
\newcommand{\bx}{\bm{x}}
\newcommand{\bs}{\bm{s}}

\begin{document}

\maketitle
\begin{abstract}

    Conjoint experiments provide an attractive way to assess the role of multiple attributes simultaneously on decision-making. However, the randomization of multiple attributes prevents understanding the causal mechanisms that, critically, depend on the relationship between attributes---e.g., how one attribute affects the respondent's belief as to another attribute. This is because conjoint experiments recover controlled effects whereas a substantively important estimand may be the total or indirect effect of one attribute. Unfortunately, existing experimental designs for conjoint experiments cannot estimate these effects. We provide an alternative framework that requires one additional, simple experiment to learn the relationship between attributes among respondents alongside the standard assumptions for causal mediation. Estimation of the relevant effects can be done in a doubly robust fashion using machine learning methods. We illustrate this by conducting a pre-registered experiment on candidate choice and disentangle the effect of different attributes by understanding their mediation through the candidate's party.


    \vspace{1em}

    \ifnum\blind=1
        \vspace{1em}\noindent\textbf{Word Count:} 10974
    \else 
        ~
     \fi
     
\end{abstract}
\thispagestyle{empty}

\clearpage
\setcounter{page}{1}
\doublespacing

\section{Introduction}

Conjoint experiments are a popular method in social science, where they are used to randomize multiple (often many) different attributes (treatments) to create realistic profiles and assess respondents' preferences over them. This makes these experiments particularly appealing in fields that care about choices and tradeoffs over multidimensional items, such as marketing \citep{green1990conjoint, bradlow2005current, allenby2019economic} and political science \citep{hainmueller2014causal,rao:14,abramson2022we,bansak2023using}. Some influential and popular applications in political science include the types of immigrants that respondents would prefer to admit to the United States \citep{hainmueller_2015_ajps} and the types of candidates that are preferred in an election \citep{ono2019contingent,schwarz2022have}. In marketing, these experiments are widely used to assess preferences over consumer goods \citep[e.g.,][]{rao:14, netzer2011adaptive, eggers2021choice}.

A critical decision when designing a conjoint experiment is which attributes to include. Consider an experiment in the United States where voters are asked to assess hypothetical candidates for office, and suppose one is especially interested in estimating the effect of race on the probability of selecting a candidate. Researchers often wish to make the decision as realistic as possible. In many jurisdictions, the party of the candidate does not appear on the ballot for municipal elections. Thus, a conjoint experiment that seeks realistically to mimic what those voters actually see should not include candidate party. However, if the experiment finds, for example, an effect of the candidate's race when the candidate's party is not included, the interpretation of that effect and attributing it to the candidate's race is delicate. Even though candidate party is not explicitly provided, it likely still plays a role: If respondents know about a candidate's race, it is reasonable to suspect that they will draw inferences about party and this may inform their decision-making. Thus, it is possible that the effect of race flows entirely through the changes in belief it causes about the candidate's party. Indeed, existing work has found that providing party labels leads to different estimated effects in American and comparative contexts  \citep[e.g.,][]{kirkland_candidate_2018,kuriwaki2025winning}.

Examples of respondents drawing inferences about omitted factors are common across all types of conjoint experiments in political science and other disciplines such as marketing. Studies manipulating race on a CV might cause respondents to draw inference about social class (e.g., \citealt{bertrand2004emily,butler2017empirical}). Manipulating gender might cause inferences about party or ideology \citep{abramson2024gender}; randomizing party might change beliefs about policy positions \citep{orr2020policy}; manipulating the regime of a state (democratic or not) might change beliefs about its racial makeup \citep{rathbun2025separate,tomz2026race}. This issue is closely related to the problem of bundled treatments and information equivalence in survey experiments more generally \citep[e.g.,][]{sher2006information, hainmueller2014causal,butler2017empirical,Dafoe_Zhang_Caughey_2018}.

In many conjoint experiments, however, \emph{any} choice of attributes will enable the researcher---under mild assumptions---to identify and estimate some causal quantity (e.g., the effect of candidate race on vote choice, marginalizing over all other attributes; \citealt{hainmueller2014causal}). Rather, the crux of the problem is one of interpretation as including different factors changes the quantities of interest that can be identified. Imprecision about the causal quantity at hand can lead to starkly different interpretations of the same study (see, e.g., Appendix~\ref{app:additional_illustration} for an illustration in recent published work). Unfortunately, given the standard way that conjoint experiments are designed, it is generally not possible---even with strong assumptions---to recover the quantity identified from an experiment where, say, $T$ and $S$ are included from one where $T$, $S$ and $M$ are included. Thus, researchers face a difficult choice in deciding which factors to include.

In this paper, we argue that the framework of causal mediation can help guide researchers to design studies to elicit causal quantities that are of substantive interest as well as allowing them to recover more subtle quantities of interest---albeit at the expense of making stronger assumptions. The benefit of mediation is to allow for the decomposition of the effect of some treatment $T$ on an outcome $Y$ to be decomposed into ``direct'' and ``indirect'' effects \citep{pear:01,imai_2011_apsr,vanderweele_2014_epidem}. Indirect effects are of particular interest theoretically as they capture the portion of the effect of $T$ on $Y$ that flows through $T$'s effect on a mediator $M$, whereas direct effects measure the portion of the effect of $T$ that exists without manipulating $M$. However, existing conjoint experiments are simply unable to estimate these effects---even given the standard (strong) assumptions required for mediation. The key limitation is that most conjoint experiments do not collect information about how beliefs in some non-randomized factor are affected by the randomized treatments (though see, e.g., \citealt{Dafoe_Zhang_Caughey_2018,kuriwaki2025winning}), perhaps due to the concern that asking respondents about their beliefs about a candidate's party simultaneously to asking their vote choice would affect their responses.

To tackle this, we propose a new experimental design that is a straightforward extension of standard conjoint experiments that allow researchers to perform mediation analysis. In brief, the researcher combines a (standard) conjoint experiment where both the treatment and mediator are randomized with an additional experiment where a (separate) set of respondents reveal their beliefs as to values of the mediator given the randomized profiles---but where they are not asked about the outcome. We show that, under the standard assumptions in mediation \citep[e.g.,][]{imai:keel:yama:10,forastiere_2018_biomet} and conjoint experiments \citep[e.g.,][]{acharya_2018_polan}, one can identify all of the relevant mediation quantities. We also derive a novel sensitivity test to (some) violations of the assumptions that can be implemented if one runs a standard conjoint experiment where the mediator is not included. To help researchers implement this design, we discuss in detail how quantities of interest change when a mediator is included and provide guidance on designing conjoint experiments in the presence of mediators.

To estimate these mediation effects in a credible fashion, one must estimate multiple high-dimensional regressions that control for a rich set of pre-treatment covariates alongside the randomized attributes in the conjoint. We apply ideas from existing work that combines machine learning and causal inference (e.g., \citealt{semenova2021debiased,farbmacher_2022_econmet,ratkovic2023relaxing,kennedy2024semiparametric}) to analyze the output from our experimental design. One can estimate average effects in a robust fashion and obtain interpretable summaries of heterogeneous effects, both with valid standard errors.

To illustrate the insights gained from applying mediation analysis, we conduct a pre-registered replication of \cite{kirkland_candidate_2018}'s study of candidate choice for mayoral elections, where they compare effects from a study where party is not provided against one where party is provided and randomized. We show that the difference between the effects in those experiments does not always correspond to the existence of an indirect effect. On the one hand, the effect of candidate race is strongly mediated through party (i.e., an indirect effect) where Democratic respondents exhibit large positive effects towards non-white candidates and Republican respondents exhibit large negative effects. On the other hand, for political experience, both types of respondents show small indirect effects---but large and positive direct effects---towards candidates with political experience.

\section{Information and Designing Conjoint Experiments}

Before turning to a formal analysis of mediation in conjoint experiments, it is useful to explain more precisely what changes when a mediator---such as candidate party---is included and randomized in a conjoint experiment. Imagine the researcher is primarily interested in the effect of some variable $T$, say, candidate race, on an outcome $Y$, say, candidate choice. One could design a study where only $T$ was included in the study; we refer to this as a $Y(T)$-experiment. With mild assumptions, one can estimate the \emph{total} causal effect of $T$ on the outcome $Y$ \citep{imai_2013_jrss,hainmueller2014causal,acharya_2018_polan}.\footnote{This is often known as the average marginal component effect if there are other attributes that are marginalized over, see Section~\ref{sec:theory}.} 

However, an objection to interpreting this as a theoretically interesting effect of $T$ is that there may be some mediator $M$---candidate party, for example---that is important and primed by $T$. As our empirical example confirms, Black candidates are strongly believed to be Democrats by American survey respondents, and thus we conjecture that party is likely an important mediator for the effect of race. This leads to an important substantive question in a conjoint that only includes candidate race: Do respondents change their vote choice because of a candidate's race ($T$) \emph{directly}? Or, alternatively do respondents, being told about the candidate's race, update their beliefs about the candidate's partisanship ($M$) and $M$ itself is the reason for their decision? A strong and reasonable objection to a $Y(T)$-experiment, then, is that if one provided information about party, the effect of race would disappear. Thus, a $Y(T)$-experiment that omits party ($M$) would ``mis-attribute'' a causal effect to race when it is wholly operative through race's effect on party.

The natural response then is to try to block many of these other channels by specifically providing information about them  (see also \citealt{Dafoe_Zhang_Caughey_2018}'s ``covariate control''); this is often easy to do in conjoint experiments as, by construction, respondents are evaluating profiles that the researcher directly controls. In this stylized example, one would randomize both $T$ and $M$; we refer to this as a $Y(T,M)$-experiment. The effect of $T$ that is obtained is often known as a controlled direct effect if $M$ is fixed at some value or an average marginal component effect (AMCE) if one averages these effects over some distribution of $M$ \citep{hainmueller2014causal, acharya_2018_polan, de2022improving}. If this $Y(T,M)$-experiment showed no effect of $T$, it would vindicate the objection to the $Y(T)$-experiment. 

However, there is an underappreciated cost to including both $T$ and $M$ in the experiment. First, the interpretation of the controlled direct effect and AMCE is subtle (see, e.g., \citealt{abramson2022we,bansak2023using,ganter2023identification}). In terms of understanding causal mechanisms, researchers might err by mistakenly interpreting it as either a total effect, ignoring the pathways that are blocked by fixing $M$ at some value, or as a natural direct effect, ignoring the role of the manipulation of $M$ \citep{pear:01, vanderweele_2014_epidem}.\footnote{The distinction between a controlled direct effect and a natural direct effect is subtle \citep{pear:01}; Section~\ref{sec:theory} and Appendix~\ref{app:example_direct} provide illustrations in terms of an idealized experimental procedure.} Second, if $T$ is the variable of primary \emph{theoretical} interest and $M$ is included mostly to address the objection to a $Y(T)$-experiment, a $Y(T,M)$-study alone will again mis-attribute and underestimate the importance of $T$ by ignoring a possibly important channel ($M$) through which it operates. \cite{FuLi2026} provides interesting evidence for this, as they note that observed treatment effects in conjoint experiments appear to decline as one includes more attributes (Figure 2). Including too many attributes also may undermine the realism of the experiment if one provides too much information compared to what respondents actually have when actually making decisions. 

Unfortunately, it remains unclear how to address this tension when designing conjoint experiments despite good general guidance on other aspects of design (e.g., \citealt{hainmueller2014causal, bansak2018number, de2022improving}). For example,  \citet{rao:14}'s well-cited textbook on conjoint experiments devotes only two pages to which attributes to include and says it is ``as much an art as a science'' (p. 43). Explicitly thinking about the problem using causal mediation can provide guidance.

First, one should think through how information is likely to flow in the real world scenario the conjoint experiment is designed to mimic; in our case of candidate choice in municipal elections, we are guided by the fact that many elections in the United States do not provide party label on the ballot. The candidate party $M$ is a variable that is downstream from more observable quantities (age, gender, etc.) and thus is an important potential mediator. Second, it is important to pin down the correct theoretical quantities of interest. If a mediation quantity (e.g., an indirect effect) is of primary interest, then we provide a framework to estimate it---with the caveat that stronger assumptions that cannot be guaranteed by design are required to do so. Even if the quantity of interest is not one that requires mediation to obtain, our approach can still be of use as a robustness test (e.g., is there evidence of a direct effect even after accounting for $M$; e.g., a penalty for gender after accounting for ideology; \citealt{abramson2024gender}). It can also help thread the needle of having a ``realistic experiment'' (i.e., where many things, including $M$ are randomized) and obtaining a meaningful quantity of interest (e.g., the total effect of $T$ accounting for its effect on $M$).

\subsection{Challenges for Existing Designs}

Even once the causal quantities are pinned down, however, estimating them from conjoint experiments is challenging. The key limitation is that both the $Y(T)$ and $Y(T,M)$-studies are missing key information---how does $T$ affect $M$. This is because conjoint experiments very rarely measure the value of the mediator $M$ given a value of $T$ (e.g., what party did a respondent think a hypothetical candidate was likely to be).\footnote{\cite{kuriwaki2025winning} is a notable exception (see Section~\ref{sec:extensions}). In traditional parallel designs \citep{imai_2013_jrss}, $M$ is explicitly measured in the $Y(T)$-experiment.} Thus, even if one was willing to make standard assumptions in causal mediation that allow for direct and indirect effects to be identified from observational data \citep{imai_2011_apsr,vanderweele_2014_epidem}, the lack of a measured mediator prevents such an analysis.

One approach to sidestep this problem is to compare the total effect from a separate $Y(T)$-experiment against the effect of $T$ estimated from a $Y(T,M)$-experiment \citep{imai_2013_jrss,acharya_2018_polan}. Unfortunately, even with access to both $Y(T)$ and $Y(T,M)$ studies, one must make strong assumptions to recover the indirect effect.\footnote{One approach is to assume away interactions between treatment and mediator for all respondents \citep{imai_2013_jrss}, see Appendices~\ref{app:alt_mediation} and~\ref{app:no_interaction}.} Without such assumptions, the difference in estimated effects between the two studies is a combination of the indirect effect and an interaction effect between the treatment and mediator \citep{vanderweele_2014_epidem,acharya_2018_polan}. In fact, a part of a recent disagreement in interpreting a conjoint experiment turns on the distinction between this eliminated effect and a ``proper'' indirect effect (\citealt{rathbun2025separate,tomz2026race}; see Appendix~\ref{app:additional_illustration} for details). As our empirical illustration (Section~\ref{sec:empirical}), it is highly variable whether the eliminated effect does, in fact, correspond to a meaningful indirect effect.

Our proposed design tackles this problem by explicitly gathering that missing information on how $T$ affects $M$ from survey respondents. As Appendix~\ref{app:survey_details} details, this experiment looks exactly like a standard conjoint, but the respondents are asked to provide their ``best guess'' as to the value of the mediator instead of revealing their choice ($Y$). In conjunction with a $Y(T,M)$-experiment (where the mediator is randomized), standard mediation assumptions are sufficient to identify all indirect, direct, and total effects. If resources permit, we also encourage researchers to conduct the $Y(T)$-experiment as this allows one to report the eliminated effect that is identified under weaker assumptions. As we discuss in Section~\ref{sec:partial_test}, having all three experiments also lets one conduct a  sensitivity analysis to partially test violations of the assumptions.

\section{Causal Mediation in Factorial Experiments}\label{sec:theory}

We consider first a $2^J$ factorial setting where we observe data on $N$ units, i.e., where each unit $i \in \{1,,\dots,N\}$ is assigned $J$ binary factors. One factor, $T_i \in \{0,1\}$, is designated as the (focal) ``treatment,'' another, $M_i \in \{0, 1\}$, as the ``mediator,'' and the rest, $\bSi\in \{0, 1\}^{J-2}$, as ``auxiliary treatments.''\footnote{We focus here on the binary case for exposition; all factors may have multiple levels, as in our empirical example (Section~\ref{sec:empirical}). Section~\ref{sec:multiple_mediator} discusses the case of multiple mediators.}  These labels are based purely on researcher interest; treatment, mediator, and auxiliary treatments are not handled differently in the design of the $Y(T,M)$-experiment and are all randomized. To align with our empirical example, we consider the choice of candidate for mayor: $T_i$ is whether the candidate is white (0) or Black (1); $M_i$ is whether the candidate is a Republican (0) or a Democrat (1); $\bSi$ contains other information such as the age of the candidate, their prior political experience and their gender. We are not interested in mediation through $\bSi$. We also observe pre-treatment covariates $\bXi$ for each individual respondent.

Following standard work on causal inference and mediation (e.g., \citealt{pear:01,vanderweele_2014_epidem}), we define the potential outcome as a function of the treatment, mediator, and auxiliary treatments $Y_i(t,m, \bs)$ and define the observed outcome $Y_i$ as corresponding to the potential outcome associated with the observed $T_i$, $M_i$ and $\bSi$. We also define potential outcomes for the mediator $M_i(t, \bs)$ and similarly connect the observed mediator $M_i$ to the corresponding potential outcome. This makes standard assumptions such as the Stable Unit Treatment Value Assumption (SUTVA) and composition (e.g., \citealt{vanderweele_2014_epidem}) that are formalized in Appendix~\ref{app:SUTVA}. 

It is common for mediation analysis to define the key quantities of interest using an additional potential outcome $Y_i(t, M_i(t', \bs), \bs)$ where the mediator value is determined based on an intervention on the value of $t'$ but holding $t$ at some other value: In our example above, consider $Y_i(1, M_i(0, \bs),\bs)$: Fixing $\bs$, this represents the potential outcome that would obtain for person $i$ if one recorded the party they associate with a white candidate ($M_i(0, \bs)$) and then examined whether they would vote for a Black candidate with this party affiliation, i.e., $Y_i(1, M_i(0,\bs), \bs)$. These are known as ``nested'' counterfactuals and can be somewhat challenging to interpret.\footnote{\cite{vanderweele_2014_epidem} analyzes mediation without nested counterfactuals; Appendix~\ref{app:alt_mediation} provides details.} Thought experiments like the one above make this somewhat easier in the conjoint experiment given that one can directly set $M$ and $T$. More generally, however, if $t \neq t'$, these potential outcomes are generally unobservable in a single experiment as they depend on two different treatment values assigned to the same unit; thus, they are often known as ``a priori'' or ``cross-world'' counterfactuals \citep{forastiere_2018_biomet}. 

Table~\ref{tab:QOI} defines our main quantities of interest using this notation. The fundamental building block is the average nested potential outcome $\alpha(t,t')$ that averages across units and auxiliary treatments $\bs$; it is the direct analogue to the average potential outcome in an experiment with only a single treatment and, when $t = t'$, it exactly corresponds to the marginal mean from a $Y(T)$-experiment \citep{leeper2020measuring}.

\begin{table}[!htbp]
\caption{Mediation Quantities of Interest}\label{tab:QOI}
\vspace{-1em} 
\centering
\[
\begin{array}{l c r@{} c@{} l}
\hline\hline
\text{Quantity} & \text{Abbreviation} & & & \text{Definition} \\
\hline \\[-1em]
\multirow{2}{*}{\text{\shortstack{Average Nested\\Potential Outcome}}}& & 
\multirow{2}{*}{$\alpha(t,t')$} & \multirow{2}{*}{$=$} & \multirow{2}{*}{$\E\left[Y_i(t,M_i(t',\bSi), \bSi)\right]$} \\[0.25em]\\[0.25em]\hline\\[-0.75em]
\text{Total Effect} & \text{AMCE} &
\tau & = & \E\left[Y_i(1,M_i(1, \bSi), \bSi) - Y_i(0, M_i(0, \bSi), \bSi)\right] \\[0.25em]\hline\\[-0.75em]
\multirow{2}{*}{\shortstack{Average Marginal\\Indirect Effect}} & \text{AMIE} & 
\delta(t) & = & \E\left[Y_i(t,M_i(1, \bSi), \bSi) - Y_i(t, M_i(0, \bSi), \bSi)\right] \\[0.25em]
 & \text{MAMIE} & 
\bar{\delta} & = & \frac{1}{2} \left(\delta(0) + \delta(1)\right) \\[0.25em]\hline\\[-0.75em]
\multirow{2}{*}{\shortstack{Average Marginal\\Direct Effect}} & \text{AMDE} & 
\xi(t) & = & \E\left[Y_i(1,M_i(t,\bSi), \bSi) - Y_i(0, M_i(t, \bSi), \bSi)\right] \\[0.25em]
 & \text{MAMDE} & 
\bar{\xi} & = & \frac{1}{2} \left(\xi(0) + \xi(1)\right) \\[0.25em]
\hline\hline
\end{array}
\]\vspace{-1em}
\end{table}

The total effect of treatment, $\tau$, is exactly the average marginal component effect from a $Y(T)$-experiment (AMCE; \citealt{hainmueller2014causal}). The average marginal indirect and direct effects (i.e., AMIE, $\delta(t)$; AMDE, $\xi(t)$) mirror the standard mediation indirect and direct effects \citep{imai_2010_psychmeth}, but we adapt their name to explicitly note that they are defined conditional on marginalizing out auxiliary treatments $\bs$.\footnote{We focus on averaging over the empirical distribution of $\bSi$ in our main analysis; Appendix~\ref{app:mediation_fmla} extends this to alternative distributions (e.g., \citealt{de2022improving}).}
One must take care to interpret these marginal effects for the same reason that care is needed when interpreting the AMCE \citep{abramson2022we,bansak2023using}. All of these effects can be defined as linear combinations of the average nested potential outcomes. For example, the AMIE $\delta(t) = \alpha(t,1)-\alpha(t,0)$. When $t$ has many levels, it can be cumbersome to look at all $\delta(t)$ or $\xi(t)$; Table~\ref{tab:QOI} thus defines a marginalized version of each. For example, $\bar{\delta}$ averages together $\delta(t)$ across all levels of $t$---hence marginalized AMIE or MAMIE. It is a standard result that $\tau = \bar{\delta} + \bar{\xi}$ (e.g., \citealt{imai_2011_apsr}). This shows that $\tau$ is also a linear combination of all four $\alpha(t,t')$.

The conjoint setting makes it natural to consider an idealized experiment for how one might elicit these effects: Focusing on indirect effect $\delta(0)$, imagine one solicited the following information: (a) What party would you associate with a Black candidate with other attributes $\bs$? This reveals potential outcome $M_i(0)$. (b) What party would you associate with a white candidate with other attributes $\bs$? This reveals potential outcome $M_i(1)$. The individual-level indirect effect is obtained by comparing how the respondent would view two white candidates ($T_i=0$) with the party provided in (a) versus the party provided in (b), i.e. $Y_i(0,M_i(1),\bs)-Y_i(0,M_i(0),\bs)$. Averaging over individuals and auxiliary treatments $\bs$ obtains the average marginal indirect effect $\delta(0)$. This isolates the effect of race through its change in the distribution of beliefs on party---holding constant the race of the candidate when respondents are asked to evaluate a candidate. The average marginal direct effect $\xi(0)$ has a similar interpretation, see Appendix~\ref{app:example_direct} for an analogous stylized experimental design. This isolates the effect of candidate race on the outcome, holding constant the distribution of beliefs about party induced by white candidates. There is thus a ``descriptive'' logic to these direct and indirect effects versus a ``prescriptive'' logic of controlled direct effects where race and party are set by the experimenter \citep{pear:01}.

The key challenge for causal mediation is that the average marginal direct and indirect effects ($\delta(t)$ and $\xi(t)$) are not identified by randomization alone as they involve the cross-world counterfactuals, i.e., $Y_i(0, M_i(1), \bs)$, discussed above. To address this, assumptions must be made. While the most common assumption is sequential ignorability \citep[e.g.,][]{imai_2010_psychmeth}, we prefer to state our assumption in terms of principal strata \citep{fran:rubi:02}. \citet{forastiere_2018_biomet} discusses the connection between these assumptions. Principal strata are most commonly encountered in political science in the context of instrumental variables (e.g., \citealt{angrist1996identification}, though see \citealt{hartman2024improving} for different use). However, they can also be applied to mediation analysis \citep{forastiere_2018_biomet}. 

We define the set of principal strata as all mappings between treatment and mediator; fixing the auxiliary treatment at $\bs$, there are four principal strata in the binary treatment and binary mediator setting, listed below where $G_i = \{M_i(0,\bs), M_i(1,\bs)\}$ and noting the connection to instrumental variables for illustration:
\begin{itemize}
    \item $G_i = \{0,0\}$: Candidates of either race are assumed to be Republicans (i.e., ``never-takers'')
    \item $G_i = \{0,1\}$: White candidates are assumed to be Republicans; Black candidates are assumed to be Democrats (i.e., ``compliers'')
    \item $G_i = \{1,0\}$: White candidates are assumed to be Democrats; Black candidates are assumed to be Republicans (i.e., ``defiers'')
    \item $G_i = \{1,1\}$: Candidates of either race are assumed to be Democrats (i.e., ``always takers'')
\end{itemize}
In total there are $4^{2^{J-2}}$ principal strata, as there are four principal strata for each $\bs$ combination. The principal strata framework can help us understand the direct and indirect effects more intuitively: There are some strata $\mathcal{G}$ where the treatment affects the mediator, i.e., $M_i(1,\bs) \neq M_i(0,\bs)$. The average marginal indirect effect is a weighted sum of the size of the strata in $\mathcal{G}$ and particular average marginal component effects within those strata. A related characterization applies to the average marginal direct effect; Appendix~\ref{app:decompose_PI} formalizes this precisely.

It is common to interpret the principal strata $G_i$ as a some pre-treatment characteristic of the respondent and thus, if it were \emph{known}, one could adjust for it in analysis by including it in $\bXi$ \citep{fran:rubi:02}. A key assumption related to principal strata, therefore, is that the potential outcomes $Y_i(t,m,\bs)$ are conditionally independent of $G_i$ given $\bXi$. Formally, 

\begin{assumption}[Principal ignorability]
\begin{equation}\tag{A1}\label{eq:A_PI}
    Y_i(t,m,\bs) \ind G_i | \bXi.
\end{equation}
\end{assumption}

Assumption~\ref{eq:A_PI} is known as principal ignorability and, like sequential ignorability, can be a rather strong assumption \citep{forastiere_2018_biomet}. For example, setting aside covariates and auxiliary treatments, it would require that, say, ``compliers'' (i.e., people who believe that race and party are associated in the conventional way) are not systematically more in favor of, say, Black Democratic candidates than individuals in other principal strata. Fortunately, it can be made more credible when covariates are included. If we condition on, say, the respondent's race, Assumption~\ref{eq:A_PI} would only require that \emph{among} Black respondents there is no systematic relationship between $G_i$ and the potential outcomes.

Importantly, Assumption~\ref{eq:A_PI} has some partially testable implications given our experimental setup. First, note that if we had a perfect predictive model for the mediator, i.e., $M_i$ is perfectly predicted by $T_i$, $\bXi$ and $\bSi$, then this implies that principal strata membership $G_i$ is perfectly predictable based on $\bXi$ and ensures that Assumption~\ref{eq:A_PI} holds. If, empirically, the predictive model for $M_i$ is very strong, there is limited scope for violations of the assumption to occur. We thus encourage researchers to collect a rich set of $\bXi$ and use a high-quality predictive model; we return to this point in Sections~\ref{sec:estimation} and~\ref{sec:partial_test}.

Second, this assumption is stronger than is needed; strictly, we only need for \emph{certain} principal strata to be conditionally independent of the potential outcomes for estimation of the cross-world conditional expectations, e.g., $\E[Y_i(t,M_i(t'),\bs)]$ for $t \neq t'$. However, a benefit to making this stronger assumption is that it leads to a testable implication discussed in detail in Section~\ref{sec:partial_test}: Does the modeled ``within-world'' counterfactual, i.e., $\alpha(t,t)$, align with one that can be non-parametrically estimated using mild assumptions from a $Y(T)$-experiment? If so, this adds credibility to the principal ignorability assumption.

\subsection{Experimental Design and Identifying Assumptions}

In our experimental design, we split subjects randomly into three groups. The first group are presented randomized value of $T_i$, $M_i$, $\bSi$, and their $Y_i$ value is measured, i.e., a $Y(T,M)$-experiment; this is the standard conjoint experiment that is conducted by researchers. The second group is shown a randomized $T_i$ and $\bSi$ and their $M_i$ is observed, i.e., a $M(T)$-experiment. The third group has only $T_i$ and $\bSi$ randomized and $Y_i$ is measured. We use $A_i$ to indicate the study into which a person was assigned, i.e. $A_i \in \{0,1,*\}$ where $0$ indicates the $M(T)$-study, 1 indicates the $Y(T,M)$ study and $*$ indicates the auxiliary $Y(T)$-experiment that we will discuss in Section~\ref{sec:partial_test}. The use of different randomized experiments requires a second critical assumption \citep{imai_2013_jrss,acharya_2018_polan} that we define as Assumption~\ref{eq:A_ExcluRes},

\begin{assumption}[Manipulation exclusion restriction]
\begin{equation}\label{eq:A_ExcluRes}\tag{A2}
    Y_i(t,m,s, a) = Y_i(t,m,s,a'); \quad M_i(t,s,a) = M_i(t,s,a').
\end{equation}
\end{assumption}

It states that the study to which an individual is assigned does not affect their potential outcomes. \cite{acharya_2018_polan} refer to this as a ``manipulation exclusion restriction'' and note that it requires that the experiment type has no effect on the outcome. The major threat to this assumption is if explicitly providing information on $M$ (party) changes how respondents view $T$ (race) versus how race is viewed if $T$ is induced by $M$. If that occurs, it is difficult to combine mediation information to learn about the quantities of interest as the mediator, effectively, is different across the two experiments \citep{imai_2013_jrss,acharya_2018_polan}. Given this assumption, however, we may drop $a$ from the potential outcomes and only refer to $Y_i(t,m,\bs)$. 

The next three assumptions are more standard and can be satisfied by our proposed design. We condition on $A_i$ for Assumptions~\ref{eq:A_RandYTM} and~\ref{eq:A_RandMT} to note that this assumption is only relevant for the corresponding study. Note that the mediator $M_i$ is randomized in the $Y(T,M)$-experiment (i.e., $A_i = 1$), and thus Assumption~\ref{eq:A_RandYTM} holds by design.

\begin{assumption}[Ignorability and positivity of experiment type]
    \begin{equation}\label{eq:A_RandExp}\tag{A3}
        Y_i(t,\bs,m) \ind A_i | \bXi; \quad 
         M_i(t,\bs) \ind A_i | \bXi; \quad 
         \mathrm{Pr}(A_i = a | \bXi) \in (0,1).
    \end{equation}
\end{assumption}

\begin{assumption}[Ignorability and positivity in $Y(T,M)$-experiment]
    \begin{equation}\label{eq:A_RandYTM}\tag{A4}
        Y_i(t,m,\bs) \ind (T_i, \bSi, M_i) | \bXi,  A_i = 1; \quad \mathrm{Pr}(T_i = t, \bSi = \bs,  M_i = m | \bXi, A_i = 1) \in (0,1).
    \end{equation}
\end{assumption}

\begin{assumption}[Ignorability and positivity in $M(T)$-experiment]
    \begin{equation}\label{eq:A_RandMT}\tag{A5}
        M_i(t, \bs) \ind (T_i,\bSi) | \bXi,  A_i =0; \quad \mathrm{Pr}(T_i = t, \bSi = \bs | \bXi, A_i = 0) \in (0,1).
    \end{equation}
\end{assumption}

Given Assumptions~\ref{eq:A_PI}-\ref{eq:A_RandMT}, all average nested potential outcomes can be identified from observable quantities as

\begin{equation}\label{eq:mediation_fmla}
\begin{split}
    \alpha(t,t') = \E[Y_i(t, M_i(t' ,\bSi), \bSi)] &= \sum_{\bx,\bs,m} \mu_Y(t, m, \bx, \bs) \times e_m(t', \bx, \bs) \times \mathrm{Pr}(\bXi = \bx, \bSi = \bs), \\
    \mu_Y(t,m,\bx,\bs) &= \E[Y_i | T_i = t, M_i = m, \bXi = \bx, \bSi = \bs, A_i = 1], \\ e_m(t', \bx, \bs) &= \mathrm{Pr}\left(M_i = m | T_i = t', \bXi = \bx, \bSi = \bs, A_i = 0\right).
\end{split}
\end{equation}

Appendix~\ref{app:mediation_fmla} provides the derivation. This is a version of the standard ``mediation formula'' \citep{pear:01} tailored to our setting where some information is missing (e.g., $M_i(T_i,\bSi)$ in the $Y(T,M)$-experiment) and we marginalize over the empirical distribution of the auxiliary treatments $\bs$. As Table~\ref{tab:QOI} notes, all mediation quantities previously discussed can be computed by simple linear combinations of $\alpha(t,t')$. For example, the average marginal indirect effect $\delta(t) = \alpha(t,1) - \alpha(t,0)$.

\section{Robust Estimation of Nested Counterfactuals}\label{sec:estimation}

Estimating $\alpha(t,t')$ from data is challenging, however, as it depends on two high-dimensional conditional expectation functions $\mu_Y(t,m,\bx,\bs)$ and $e_m(t,\bx, \bs)$. This problem of estimating a low-dimensional parameter of interest, $\alpha(t,t')$, given high-dimensional ``nuisance'' functions is a central question for the literature on machine learning and causal inference (see \cite{ratkovic2023relaxing} for an accessible introduction for political science). The key challenge is to ensure that one can obtain consistent and asymptotically normal estimators of the causal quantities of interest even when machine learning methods are used.

A large body of work has explored how to do this for standard causal quantities such as the average treatment effect  (ATE) or average treatment effect on the treated (ATT) \citep[e.g.,][]{chernozhukov2018dml,ratkovic2023relaxing}, with some more limited work on mediation quantities \citep[e.g.,][]{farbmacher_2022_econmet,liu2024general}. Our approach is closest to \cite{farbmacher_2022_econmet}'s application of double machine learning to estimate mediation effects from observational data. However, our setting is more complex insofar as we do not ever observe a $(Y_i, T_i, M_i)$ for an individual as the natural value of the mediator $M_i = M_i(T_i, \bSi)$ is not recorded. 

To derive an appropriate estimator for our design, we use results from \cite{kennedy2024semiparametric} on influence functions to construct a doubly robust estimator of the following form. Denote $p^T_{t}(\bx, \bs, a)$ as the propensity score for seeing treatment $T_i = t$, i.e., $\mathrm{Pr}(T_i = t | \bXi = \bx, \bSi = \bs, A_i  = a)$ and $p^M_{m}(t, \bx, \bs)$ for the propensity score for a particular fixed value of the mediator being observed in the $Y(T,M)$-experiment, i.e., $\mathrm{Pr}(M_i = m | \bXi = \bx, \bSi = \bs, A_i  = 1)$. $p_a(\bx, \bs)$ is the propensity score for being in experiment $A_i = 1$.  Equation~\ref{eq:influ} provides an influence function $\psi_i(t,t')$ whose expectation recovers $\alpha(t,t')$:

\begin{equation}\label{eq:influ}
\begin{split}
\alpha(t,t') &= \E[Y_i(t, M_i(t', \bSi), \bSi)] = \E[\psi_i(t,t')] \\
\psi_i(t,t') &= \frac{\I{T_i = t, A_i = 1}}{p^T_t(\bXi, \bSi, A_i = 1) p_a(\bXi,\bSi)} \frac{e_{M_i}(t', \bXi, \bSi)}{p^M_{M_i}(t,\bXi, \bSi) }\left[Y_i - \mu_Y(t,M_i,\bXi, \bSi)\right] + \\ &\frac{\I{T_i = t',A_i=0}}{p^T_{t'}(\bXi, \bSi, A_i = 0)\left(1-p_a(\bXi, \bSi)\right)} \left[\begin{split}&\mu_Y(t,M_i, \bXi, \bSi) \\&- \sum_m \mu_Y(t,m, \bXi, \bSi) \times e_m(t', \bXi, \bSi)\end{split} \right] + \\
&\sum_m \mu_Y(t, m, \bXi, \bSi) \times e_m(t', \bXi, \bSi).
\end{split}
\end{equation}

Appendix~\ref{app:influ} provides a derivation. Note that an expectation over the last term in this equation over $\bXi$ and $\bSi$ would recover Equation~\ref{eq:mediation_fmla}. The first two terms, therefore, ensure robustness to mis-specification of the predictive models. The first adjusts for errors in the conditional expectation of the outcome in the $Y(T,M)$-experiment, $\mu_Y(\cdot)$, and the second adjusts for errors estimating the mediator-treatment relationship using the $M(T)$-experiment, $e_m(\cdot)$. This influence function is quite similar to \cite{farbmacher_2022_econmet} but accounts for the particular features of our design and conjoint experiments.

The relevant nuisance functions, i.e., the conditional expectations---$\mu_Y$ and $e_m$---and the propensity scores---$p^T_{t}$, $p^M_m$, $p_a$---must be estimated carefully to obtain valid estimates.\footnote{In our proposed design, $p^T_t$, $p^M_m$, and $p_a$ are known so their true values could be plugged-in directly.} We use cross-fitting where we partition the data into $K$ separate folds. For each fold $k$, the data from the other folds are used to estimate the nuisance functions ($\hat{\mu}_Y$, $\hat{e}_m$, $\hat{p}_a$, $\hat{p}^T_t$, $\hat{p}^M_m$). Then, an estimate of the influence function $\hat{\psi}_i(t,t')$ is obtained for each $i$ in the $k$-th fold using Equation~\ref{eq:influ} with estimated nuisance functions fit using held-out data. This has numerous benefits including avoiding overfitting bias from estimating the machine learning models, reducing the convergence rate needed for inference to ones obtainable by flexible methods such as random forests, and allowing one to obtain properly calibrated standard errors (see, e.g., \citealt{chernozhukov2018dml,ratkovic2023relaxing,kennedy2024semiparametric} for details).

Standard theory allows us to construct a consistent and asymptotically normal estimator for $\alpha(t,t')$ using $\hat{\psi}_i(t,t')$ (see, e.g., \citealt{farbmacher_2022_econmet} for accompanying technical conditions in the mediation setting or the previous citations for general discussion):
\begin{equation}\label{eq:influ_est}
    \begin{split}
        \hat{\alpha}(t, t') = \frac{1}{N} \sum_{i=1}^N \hat{\psi}_i(t,t'), \quad \hat{\sigma}^2(t,t')&= \frac{1}{N} \sum_{i=1}^N \left(\hat{\psi}_i(t,t') - \hat{\alpha}(t,t')\right)^2, \\
        \frac{\sqrt{N}\left(\hat{\alpha}(t,t') - \alpha(t,t')\right)}{\sqrt{\hat{\sigma}^2(t,t')}} &\to^d N(0,1).
    \end{split}
\end{equation}

Recall, however, that our primary targets of interest are comparisons between $\alpha(t,t')$, e.g., $\delta(1) = \alpha(1,1) - \alpha(1,0)$. To estimate these, we note that the influence function for $\theta_a - \theta_b$ is the difference between the individual influence functions, i.e., $\psi_a - \psi_b$ \citep{farbmacher_2022_econmet,kennedy2024semiparametric}. By default, our software produces an estimate for all $\hat{\alpha}(t,t')$ and the accompanying variance-covariance matrix. Again, by standard assumptions, this is asymptotically normal and thus confidence intervals with correct coverage can be produced. Appendix~\ref{app:regestimate} proves that all quantities of interest in Table~\ref{tab:QOI} can be easily computed using simple linear regression using the $\hat{\psi}_i(t,t')$ as pseudo-outcomes.

\section{(Partially) Testing the Assumptions}\label{sec:partial_test}

The key assumptions from our model that are not guaranteed by design are principal ignorability (Assumption~\ref{eq:A_PI}) and the manipulation exclusion restriction (Assumption~\ref{eq:A_ExcluRes}). Both could be seen as rather strong, although conditioning Assumption~\ref{eq:A_PI} on $\bXi$ increases its credibility if $\bXi$ explain much of the relationship between $T_i$ and $M_i$. 

As noted above, we advocate for conducting a separate $Y(T)$-experiment, i.e., $A_i = *$. This allow us to estimate the average marginal component effect $\tau$ and $\E[Y_i(t,M_i(t,\bSi),\bSi)]$ non-parametrically, given the assumption of randomized $T$ that holds by design (i.e., a slight adaption of Assumption~\ref{eq:A_RandYTM}). Specifically, in this experiment, $\phi(t)=\E[Y_i | T_i = t, A_i =*]$ and $\tau = \phi(1) - \phi(0)$ \citep{hainmueller2014causal}. This can be estimated using the sample average. As the respondents for whom $A_i = *$ are totally separate from those used in estimating $\hat{\psi}_i(t,t')$, there is no concern with information leakage and the two estimators, $\hat{\phi}(t)$ and $\hat{\alpha}(t,t)$, are independent by design.

If the assumptions are satisfied, this would require that $\phi(t) = \alpha(t,t)$, or, equivalently that $\phi(1)-\phi(0) = \alpha(1,1) - \alpha(0,0)$. If this null hypothesis can be rejected, it suggests there is a violation of the assumption: Either, principal ignorability or manipulation exclusion restriction fail to hold; or, perhaps, the predictive model used for $\mu_Y(\cdot)$ and $e_m(\cdot)$ are inadequate and should be adjusted. If the null hypothesis cannot be rejected, this gives some credibility to the untestable assumptions.

If there is a difference between the two experiments that cannot be reconciled by improving the predictive model, it may be useful to examine if there are subsets of the data for which the null hypothesis cannot be rejected. For example, if it turns out that for Black respondents (denoted by $W_i = 1$), we cannot reject the null of $\E[Y_i |W_i=1,T_i = t, A_i = *]=\E[\hat{\psi}_i(t,t')|W_i=1,A_i \in \{0,1\}]$, then the assumptions are more credible for those respondents. Thus, that sub-group effect may be viewed as more reliable. Similarly, in the case where $T_i$ or $M_i$ have multiple values (e.g., candidates are assigned to be white, Black, Hispanic or Asian; the candidate could be a Democrat, Republican, or Independent), it could be that the assumptions only holds for a subset of potential outcomes. That would suggest restricting the analysis to only those profiles where the falsification test holds, i.e., only consider $\alpha(t,t')$ where one cannot reject the null of $\phi(t) = \alpha(t,t)$ and $\phi(t') = \alpha(t',t')$.

Appendix~\ref{app:sensitivity} uses this equivalence to design a novel sensitivity analysis that can be applied in our doubly robust framework, adapting existing work (\citealt{robins2000sensitivity,tchetgen2012semiparametric}). In brief, it uses the measured differences between $\phi(t)$ and $\alpha(t,t)$ to estimate the bias that comes from violating principal ignorability and/or the manipulation exclusion restriction. It then uses this bias estimate to adjust the reported mediation quantities to examine if the results remain robust.

\section{Extensions to the Basic Model}\label{sec:extensions}

Applied examples of conjoint experiments often go beyond the setting we formalized above. Some extensions, e.g., allowing the mediator or treatments to have multiple levels, are straightforward. This section discusses some additional extensions to make our framework usable alongside most standard conjoint experiments.

\subsection{Adjusting the Experimental Design}\label{sec:alt_exp}

In the design discussed above, the researcher simultaneously conducts a $Y(T,M)$ experiment and a $M(T)$ experiment, randomly allocating respondents to each study. This has benefits insofar as it means there is independence between the respondents in each study and thus certain assumptions, stated above, hold by design. 

One adjustment of the design allows for our method to be applied in a wider set of circumstances. Imagine that an existing study had conducted a $Y(T,M)$-experiment and the researcher wished to understand mediation in that study without re-conducting the same experiment---perhaps due to cost constraints. Using our framework, they would only need to conduct a $M(T)$-experiment, estimate $\mu_Y(\cdot)$ from the existing study and $e_m(\cdot)$ from the new study. In this setting, however, note that Assumption~\ref{eq:A_RandExp}---conditional ignorability of the experimental condition $A_i$---is not assured by design. The estimation requires modeling the propensity score ($p_a(\bx,\bs)$ in Equation~\ref{eq:influ}) and can adjust for differences between experimental samples, but careful argumentation would be needed to rule out unobserved confounding. Fielding the secondary study on a similar population (e.g., using the same vendor as the original study within a reasonable number of years as the original study) might ameliorate some concerns.

A different adjustment to the experimental design would have the same respondent participate in multiple experiments, e.g., in both the $Y(T,M)$ and $M(T)$-experiments. This would allow for improved power insofar as more data could be obtained from the same respondents. Some existing studies---specifically \cite{kuriwaki2025winning}---have already implemented a version of this design. However, re-using respondents across tasks requires careful argumentation to ensure that the assumptions are met. Specifically, the manipulation exclusion restriction (Assumption~\ref{eq:A_ExcluRes}) is perhaps less credible if a respondent has done multiple tasks in quick succession. Further, one would need to make an assumption about stability and no carry-over effects. One must assume that individuals do not change how they relate treatment to outcome (e.g., candidate choice) after having been primed by a task relating the mediator to treatment (e.g., candidate race to candidate party) and vice versa. A more mild version of no carry-over is, in fact, already rather common in conjoint experiments that have respondents perform multiple tasks \citep{hainmueller2014causal}. If one is willing to believe a strong version of a no carry-over assumption, Appendix~\ref{app:crossover} sketches how one might extend our method to use a crossover design to identify direct and indirect effects without requiring principal ignorability.

\subsection{Heterogeneous Effects}\label{sec:heteff}

The above discussion has focused entirely on the average total, direct and indirect effects. Instead, interest may be in estimating conditional direct and indirect effects, i.e., $\E[Y_i(t,M_i(1))-Y_i(t,M_i(0))|\bXi]$. To treat this formally, we consider the generalization of the conditional average treatment effect to mediation, e.g., the conditional average marginal direct effect, the conditional average marginal indirect effect, and the conditional average marginal total effect. These can be defined, as above, using $\alpha(t,t';\bx) = \E[Y_i(t,M_i(t'),\bSi) | \bXi = \bx]$ where, for example, a conditional average indirect effect is $\delta(1;\bx) = \alpha(1,1;\bx) - \alpha(1,0;\bx)$.

\cite{semenova2021debiased} provide a framework for estimating low-dimensional summaries of these heterogeneous effects using $\hat{\psi}_i(t,t')$ as they satisfy, by construction, required orthogonality conditions \citep{kennedy2024semiparametric}. Specifically, for a binary covariate $W_i \in \{0,1\}$, assumed to be a subset or deterministic transformation of $\bXi$, one might wish to compare whether the conditional effect is larger when $W_i = 1$ versus $W_i=0$ with an accompanying standard error. If we focus on the indirect effect $\delta(1)$, this could be estimated by computing the within strata average of $\hat{\psi}_i(1,1)-\hat{\psi}_i(1,0)$ and then taking the difference. Appendix~\ref{app:regestimate} shows that an equivalent way to estimate this is to run a linear regression on the $\hat{\psi}_i(1,t)$ with an interaction between an indicator variable for $t'=1$ and $W_i = 1$. Under mild additional assumptions, this is a consistent and asymptotically normal estimator of the difference in average marginal indirect effects between groups \citep{semenova2021debiased}. Appendix~\ref{app:deriv_heteff} generalizes to multiple $W_i$ using a standard (low-dimensional) linear regression to predict $\hat{\psi}_i(t,t')$. This returns a consistent estimator of the best linear approximation to the conditional effect function for any of our mediation quantities with valid confidence intervals.

\subsection{Repeated Observations and Forced-Choice Design}

It is common to extend the design in conjoint experiments to have respondents perform multiple tasks and to have respondents choose between two profiles---i.e., a ``forced choice'' conjoint. The former requires standard assumptions of stability, no carry-over and no profile order effects \citep{hainmueller2014causal} but otherwise fits into our framework. For forced choice experiments, using the popular choice-level analysis that stacks the data for each profile and clusters the standard errors accordingly can be easily implemented in our framework. Future work could extend our method to the profile-level analysis \citep{clayton2025correcting}, perhaps by differencing the treatments between profiles \citep{egami2019causal,goplerud2025estimating}.

We note that the $M(T)$-experiment can be done either as a factorial or forced choice experiment depending on what is most plausible for the application. In our case of identifying party, we believed that it was easier to solicit a ``best guess'' as to the party of a hypothetical candidate vs. asking which one was ``more likely'' to be a Democrat.

\subsection{Multiple Mediators}\label{sec:multiple_mediator}

In many settings, there are often multiple factors that might mediate the effect of treatment $T$. For example, in our setting where party is our primary mediator, scholars have suggested that policy positions or ideology may be inferred from candidate characteristics such as gender (e.g., \citealt{abramson2024gender}). Even an important mediator such as party itself may also not be the ``final'' mediator: \citet{orr2020policy} show that the effect of partisanship itself might flow through induced beliefs about policy positions. It is an open question how to extend our $M(T)$-experiment to the case of multiple and/or causally ordered mediators (e.g., as analyzed in \citealt{zhou2022semiparametric}) or with more complex causal stories. However, one approach to stay within the single mediator framework would be to define a single combined mediator from the multiple mediators (e.g., ``Democrat \& Liberal'', ``Republican \& Moderate'', ``Independent \& Conservative'', etc.) and adjust the $M(T)$-experiment accordingly. Appendix~\ref{app:multiple_mediators} sketches how one might do this. In general, we suggest that scholars define $M$ as the most important or theoretically relevant mediator to allow for the indirect and direct effects with respect to that mediator to be recovered, noting that the channels that flow through other mediators are absorbed into $T$ (or $\bm{S}$). 

\section{Application to Kirkland and Coppock}\label{sec:empirical}

We apply our method to a pre-registered replication of \cite{kirkland_candidate_2018}. Their study focuses on how voters form preferences on candidates for mayor both in the presence and absence of party information, noting that many mayoral elections in the United States are officially non-partisan and do not permit the inclusion of party information on the ballot. The original study conducts a $Y(T,M)$-experiment and $Y(T)$-experiment, in our framework, where $M_i$ is the partisanship of the candidate (Democrat, Republican or Independent), $T_i$ is the candidate's race and $\bSi$ includes age, gender, previous occupation, and political experience. 

\cite{kirkland_candidate_2018} do not perform a mediation analysis or interpret the results in those terms. If one were to interpret their analysis using our framework, they focus on both the AMCE from the $Y(T)$-experiment (a ``non-partisan election'' where party is not provided) and how this changes when compared against the AMCE from the $Y(T,M)$- experiment (a ``partisan'' election where party is provided to the respondent). While they do not describe or interpret it as such, this is equivalent to the ``eliminated effect'' \citep{vanderweele_2014_epidem,acharya_2018_polan}. As discussed above, a non-zero eliminated effect gives evidence that there is \emph{some} role of $M_i$ as a causal mechanism and implies some portion of the effect flows through either mediation or a mediated interaction, but it is not able to dis-entangle direct and indirect effects. In fact, a non-zero eliminated effect could be consistent with \emph{either} a zero or non-zero indirect effect \citep{vanderweele_2014_epidem,acharya_2018_polan}. However, we think it is reasonably common to infer that \emph{some} indirect effect exists if there is an eliminated effect, and thus we used that to guide our pre-registration (available \chref{\urlprereg}{here}). Specifically, we pre-registered the following six hypotheses around the existence and sign of indirect effects assuming that the effects shown in \cite{kirkland_candidate_2018} imply the existence of an indirect effect:

\begin{enumerate}
    \item A positive indirect effect for Black vs. white candidates among Democratic respondents.
    \item A negative indirect effect for Black vs. white candidates among Republican respondents.
    \item A positive indirect effect of political experience across all respondents
    \item A positive indirect effect of political experience for Democratic respondents but not Republican respondents.
    \item A negative indirect effect for male vs. female candidates among Democratic respondents.
    \item A positive indirect effect for police officers and small business owners (vs. educators) for Republican respondents.
\end{enumerate}

In brief, the pre-registered hypotheses suggest there are certain factors (race, gender, and occupation) for which there is an effect on voter choice through its priming of party. To test this, we conducted a study on Prolific in the fall of 2025 where approximately 4,500 subjects were recruited. We chose Prolific as it is seen as a generally high-quality platform for online survey respondents (e.g., \citealt{douglas_data_2023}). Appendix~\ref{app:survey_details} provides additional details such as balance tests of pre-treatment covariates across designs, quotas used to recruit respondents to ensure a sufficient number of Republican respondents, and information on our attention check).

These respondents were allocated randomly to one of three designs ($A_i \in \{0,1,*\}$) and each performed five comparison or rating tasks. We designated the candidate's party $M_i$ as the key mediator and elicited information on how respondents associate this with other attributes in the $M(T)$-experiment. In the $M(T)$-experiment, we showed respondents a hypothetical candidate and asked them to give their ``best guess as to whether the candidate was a Democrat, Republican, or Independent''.\footnote{We randomize the order of the options across respondents. Appendix~\ref{app:survey_instrument} provides the full survey instrument; the $Y(T)$ and $Y(T,M)$-experiments closely mirrors \cite{kirkland_candidate_2018}.} The remainder of this section analyzes the results of these experiments, focusing in particular on our $M(T)$-experiment and mediation quantities. Appendix~\ref{app:emp_eliminated} reports the effects from the $Y(T)$ and $Y(T,M)$ experiments as well as the eliminated effects. It also compares these estimates to those reported in \cite{kirkland_candidate_2018}; in general, the effects are relatively similar across studies.

\subsection{$M(T)$-Experiment}

We did not pre-register hypotheses about the $M(T)$-experiment as our primary interest is on the mediation quantities. However, as these experiments are less common in political science, we report the results here. We further note that if there were no effects of $T$ on $M$, no indirect effect could exist by definition. Thus, we encourage researchers to pre-register and report the results of the $M(T)$-experiment, as this can be identified without requiring strong assumptions (i.e., without principal ignorability), to provide evidence about how $T$ affects $M$. Figure~\ref{fig:MT} shows the average treatment effect on the belief about a candidate's party. To examine whether this relationship differs by respondent, we show the results separately by Democratic and Republican respondents; Appendix~\ref{app:full_results} provides results for the full sample and for Independent respondents. Effects for all factors (age, gender, previous occupation [``Job''], political experience [``Exp.''], race) are shown.

\begin{figure}
    \caption{Estimated Effects from $M(T)$-Experiment}\label{fig:MT}
    \includegraphics[width=\textwidth]{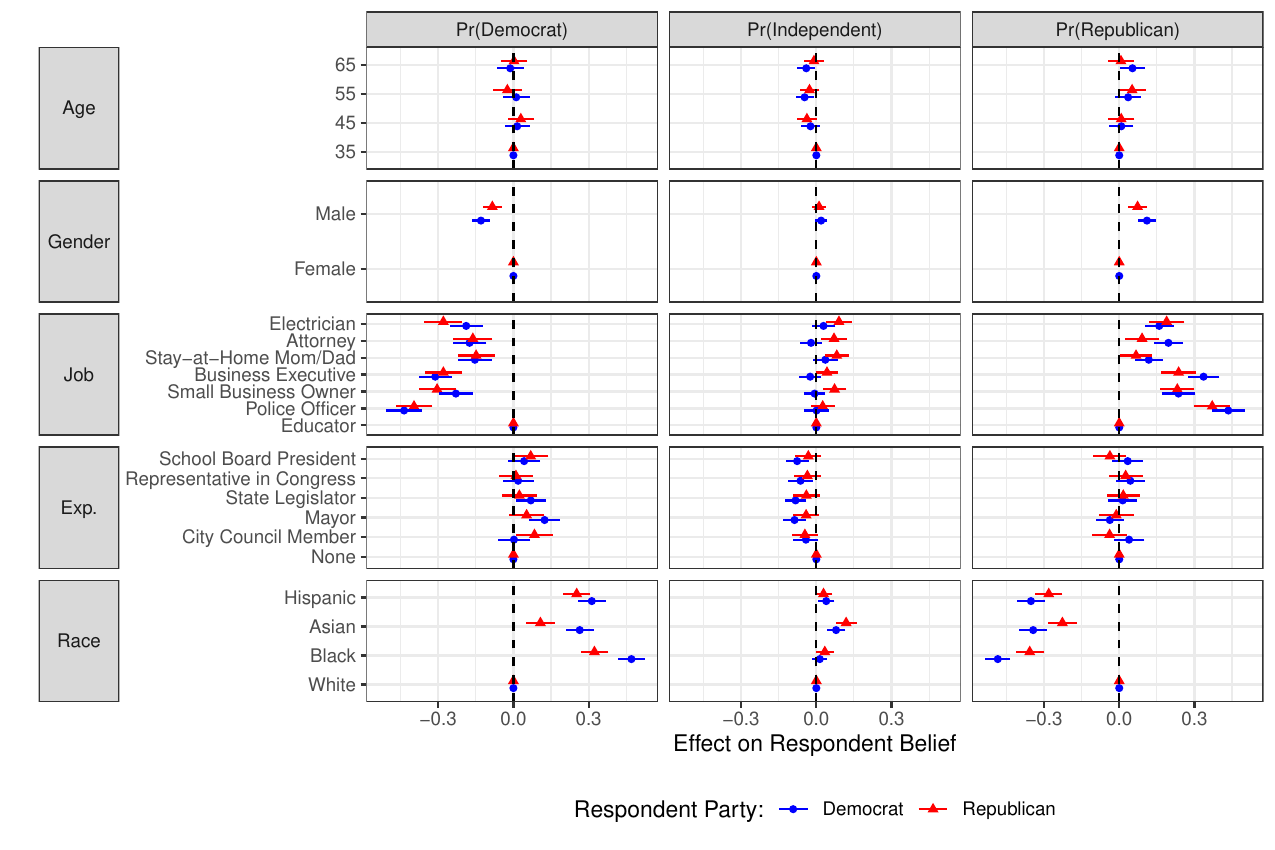}
\end{figure}

Figure~\ref{fig:MT} shows that beliefs about a candidate's party are affected by other characteristics. For example, non-white candidates are viewed as much more likely to be Democrats across all types of respondents. Further, occupation also strongly affects beliefs about candidate party. There are negative and significant treatment effects for all occupational categories versus the baseline of ``educator'' on the probability of being a Democrat (and positive and significant effects for being a Republican). Candidate gender is also associated with perceived party as male candidates are perceived as significantly more likely to be Republicans (and less likely to be Democrats). While the effects by the respondent party are generally similar, there is some evidence that Democratic respondents show stronger associations with party and other characteristics than Republican respondents. Exploring this in future pre-registered studies is an interesting area for further research.

\subsection{Mediation Analysis}

Our initial focal treatment $T_i$ is race. Following \cite{kirkland_candidate_2018}, $T_i \in \mathcal{T}$ where $\mathcal{T}=\{\texttt{White}, \texttt{Black}, \texttt{Asian}, \texttt{Hispanic}\}$. This leads to many possible average direct, indirect and total effects as different contrasts of candidate races. For simplicity, we compute the direct, indirect and total effects versus a single reference category (white). As noted in Table~\ref{tab:QOI}, we also average together the $\delta(t)$ for all values of $\mathcal{T}$ with equal probability and denote the corresponding quantity as $\bar{\delta}$; we call this the marginalized average marginal effect (MAMIE) and do the same for the direct effects to obtain the MAMDE ($\bar{\xi}$). For notational simplicity, we sometimes refer to these as indirect and direct effects in the following discussion.\footnote{Our pre-registration does not state the precise estimand (saying merely ``direct'' or ``indirect'' effect). Thus, Appendix~\ref{app:extra_a_effects} also shows the results for each $\delta(t)$ and $\xi(t)$. The results are virtually identical.} Appendix~\ref{app:survey_ML} provides details on how the cross-fitting was conducted to obtain $\hat{\psi}_i(t,t')$; in brief, we used a random forest where treatments and pre-treatment demographic covariates were included as predictors.

A benefit of our approach is that once we have fixed the key mediator as $M_i$, we can re-define the focal treatment as some other attribute (e.g., candidate political experience), repeat the cross-fitting procedure, and obtain effects for that factor as well. Figure~\ref{fig:main_results} reports the 
(marginalized) average marginal direct, indirect and total effects for all factors in the study (age, gender, previous occupation [``Job''], political experience [``Exp.''], and race). Given that five of our six pre-registered hypotheses concern effects for Democratic and Republican respondents and that our sample skews Democratic despite our pre-registered quotas on Prolific, we present all of our results in the main text for those two sub-groups.\footnote{We include partisan leaners with their respective party; using self-reported partisanship, our respondent pool was 47.2\% Democrats, 37.3\% Republican, and 15.5\% Independent, see Appendix~\ref{app:survey_demo} for details.} Appendix~\ref{app:full_results} shows the results for the entire sample and for Independent respondents, although we did not pre-register specific hypotheses for this group.

\begin{figure}[!htbp]
\caption{Estimated Mediation Effects}\label{fig:main_results}
\includegraphics[width=\textwidth]{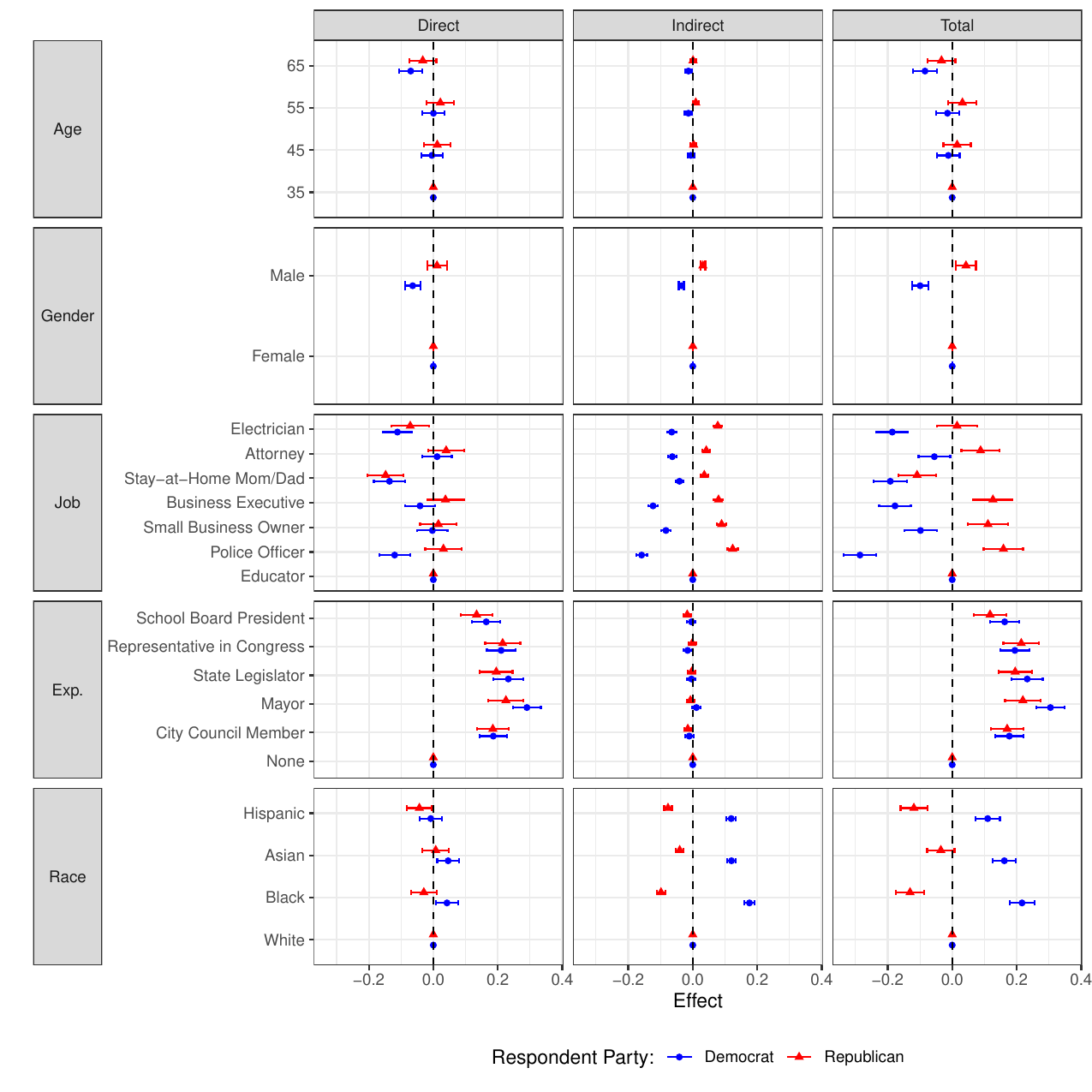}
\end{figure}

\textbf{Race:} There are sharply different indirect effects by respondent party. Republicans have a strong indirect penalty against non-white candidates because of how race affects their belief as to the candidate's party whereas Democrats have strongly positive indirect effects. This provides strong evidence for our first and second pre-registered hypotheses and aligns with our motivating example that Black candidates are heavily assumed to be Democratic. To illustrate this, our $M(T)$-experiment shows that 70\% of Black candidates are assumed to be Democrats versus 30\% of white candidates, 45\% of Asian candidates, and 55\% of Hispanic candidates. The direct effects of race are much smaller in magnitude, but are sometimes significant. For Democratic respondents, the effects for Black and Asian candidates are both significant ($p$-value of 0.017 and 0.007, respectively); for Republican respondents, the negative direct effect for Hispanic candidates is significant ($p$-value of 0.030). This shows that looking only at the $Y(T)$-experiment could be misleading as the total effects found in that study are heavily mediated through the candidate's party (see Appendix~\ref{app:emp_eliminated}).

\textbf{Political Experience (``Exp.'')}: There is considerable homogeneity across Democratic and Republican respondents, and the total effect is almost entirely through a direct effect. The indirect effects are substantively small in magnitude with only three rising to conventional levels of statistical significance.\footnote{For Democratic respondents, only Representative in Congress ($p$-value of 0.02) is significant. For Republicans, only School Board President and City Council Member ($p$-value of 0.01 for both).} This differs from our third and fourth pre-registered hypotheses regarding the role of political experience.\footnote{Our third pre-registered hypotheses is about an effect across all respondents; pooled results in Appendix~\ref{app:full_results} show large and significant direct effects. While the indirect effects are statistically significant (except for Mayor), they are similarly very small in magnitude.} As Appendix~\ref{app:emp_eliminated} shows positive eliminated effects---especially for Democratic respondents, this is illustrates that eliminated effects can be consistent with substantively small indirect effects.

\textbf{Gender:} There are indirect effects of opposing direction for Democrats and Republicans---i.e., Democrats show an indirect preference against male candidates. This supports our fifth pre-registered hypothesis. This again corresponds to the raw data in the $M(T)$-experiment where 57\% of female candidates were assumed to be Democrats versus 45\% of male candidates. Interestingly, Democratic respondents also show a negative direct effect towards male candidates. As discussed in Section~\ref{sec:multiple_mediator}, a possible explanation could be that gender signals ideology---in addition to partisanship (see, e.g., \citealt{orr2020policy,abramson2024gender}).

\textbf{Candidate Occupation (``Job'')}: We again see sharply distinct indirect effects. Democratic respondents have a negative indirect effect for all other occupations (versus educator) and Republicans have positive indirect effects. This is, again, due to the fact that educators are overwhelmingly estimated to be Democratic: 75\% of profiles in the $M(T)$-experiment with educator are believed to be Democrats. Interestingly, this factor also shows clear evidence of direct effects: Respondents of both parties show clear negative direct effects towards electricians and stay-at-home parents vs. being an educator. More strikingly, Democratic respondents also show a clear negative direct effect against police officers---in addition to their negative indirect effects. We do not find support for our sixth pre-registered hypothesis of a positive indirect effect for police officers and small business owners for Republican respondents.

\textbf{Age:} We did not pre-register hypotheses about this variable, but we report results on it for completeness. There are not many significant effects except for a negative direct effect for Democratic respondents against the oldest candidates (65) versus the youngest candidates (35) and substantively small---but significant with $p<0.05$---indirect effects for Democratic respondents for the two oldest groups (55 and 65).

\subsection{Exploratory Heterogeneous Effects}

Beyond our pre-registered hypotheses, it is interesting and important to explore for further heterogeneity.\footnote{Our pre-registration plan states that we will conduct exploratory testing for heterogeneous effects, although it does not specify how as it was submitted before the specific method was finalized.} However, we measured many respondent characteristics, and performing many sub-group analyses is both challenging to interpret and raises concerns about multiple-testing problems \citep{goplerud2025estimating}. To address this, we use the approach in Section~\ref{sec:heteff} and produce the best linear approximation to the conditional effect functions, i.e., run a linear regression to predict $\hat{\psi}_i(t,t')$ using pre-treatment covariates (age, education, ethnicity, gender, ideology, income, party ID, and political interest). We then compute the change in the predicted direct, indirect and total effects if one covariate changed (e.g., respondent party) holding all others constant. Focusing on race as the focal treatment, Figure~\ref{fig:blp_heteff} shows how the predicted average marginal direct, indirect and total effects for a Black vs. white candidate change holding other respondent-level covariates in the linear predictive model constant.

\begin{figure}[!htbp]
\caption{Exploratory Heterogeneous Effects}\label{fig:blp_heteff}
\includegraphics[width=\textwidth]{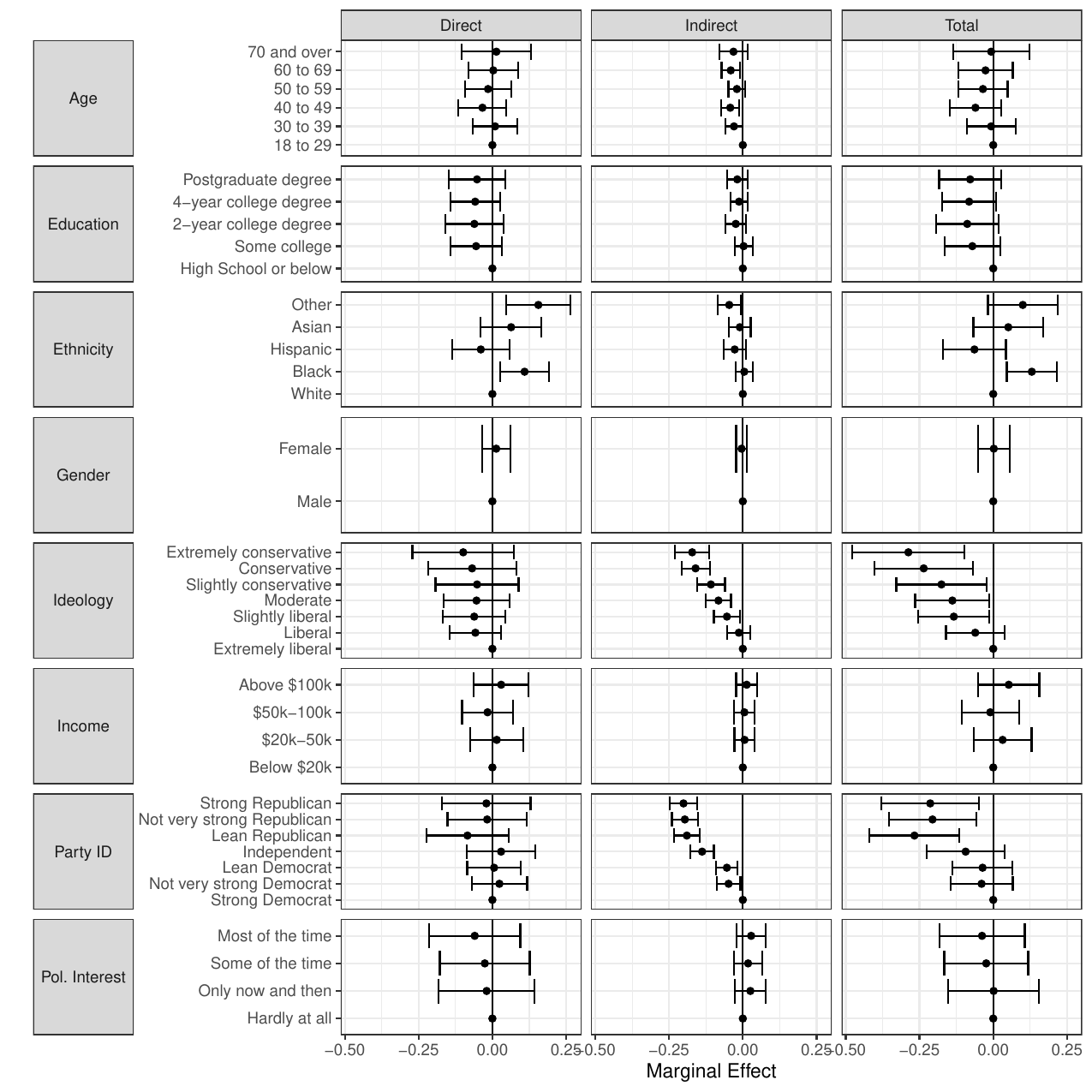}
\end{figure}

The results corroborate that the party-specific effects found in the pre-registered analysis appear robust. Controlling for a wide variety of demographic characteristics, respondents who are more conservative and more Republican both show increasingly negative indirect effects against Black candidates. Interestingly, Black respondents and respondents from other racial groups (e.g., multiracial respondents) show positive direct effects for Black candidates, controlling for many other respondent-level demographic characteristics.

\subsection{Testing the Assumptions}

Finally, as discussed in Section~\ref{sec:partial_test} and Appendix~\ref{app:sensitivity}, our separate $Y(T)$-experiment allows us to partially test whether there are violations of principal ignorability (Assumption~\ref{eq:A_PI}) and/or the manipulation exclusion restriction (Assumption~\ref{eq:A_ExcluRes}). To do this, we compute the estimates of $\E[Y_i(t,\bSi)]$---analogous to the marginal means from \cite{leeper2020measuring}---estimated using the mediation formula ($\alpha(t,t)$) and directly from the auxiliary $Y(T)$-experiment. Figure~\ref{fig:robust} presents this for Democratic and Republican respondents.

\begin{figure}[!htbp]
\caption{Testing Assumptions (Marginal Means)}\label{fig:robust}    \includegraphics[width=\textwidth]{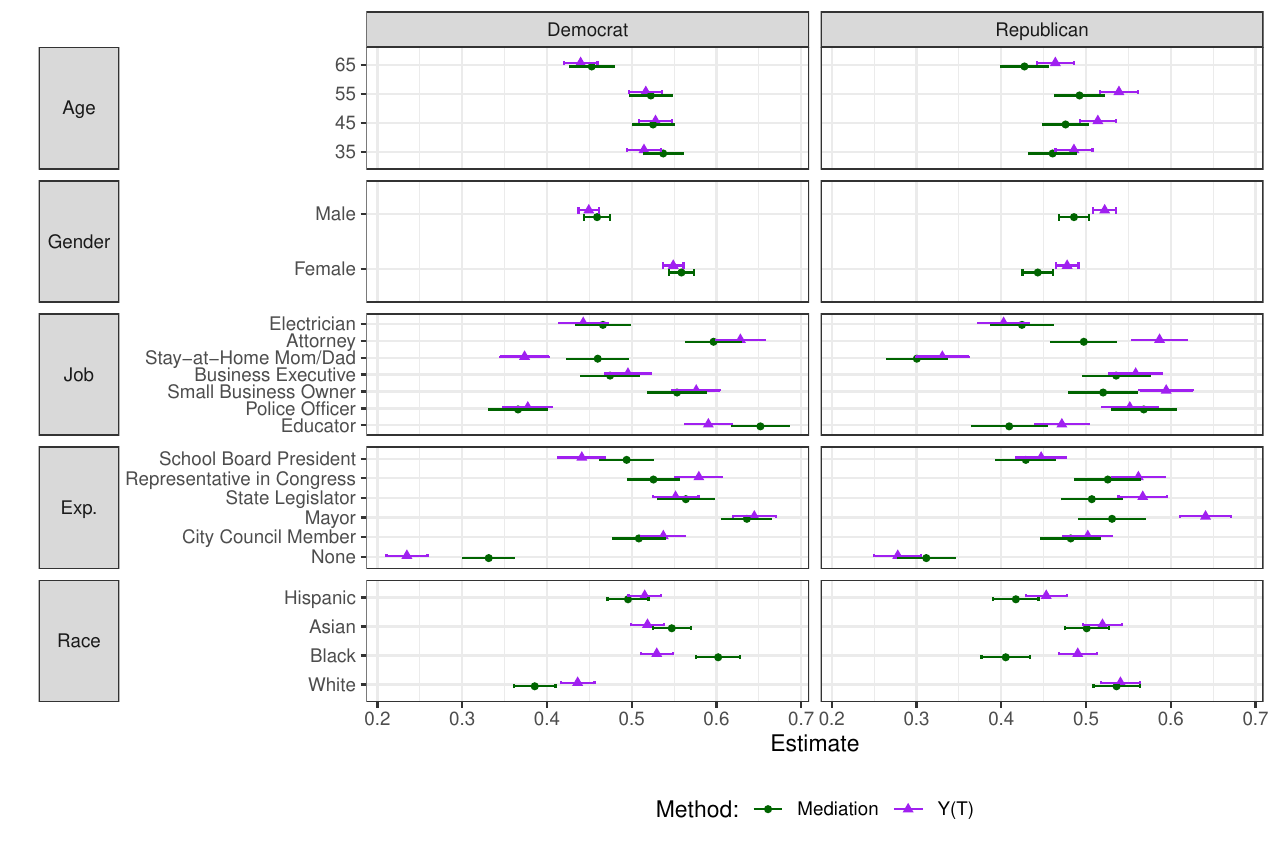}
\end{figure}

The results suggest that the assumptions vary in their plausibility for different choices of $t$ and types of respondents. While for many marginal means the estimates overlap, there are clear cases where differences exist between the $Y(T)$-experiment and the results using mediation. For Democratic respondents and for candidate race, the mediation analysis finds a lower marginal mean for white candidates (i.e., $\hat{\alpha}(t,t) < \hat{\phi}(t)$) and a larger marginal mean for Black candidates (i.e., $\hat{\alpha}(t,t) > \hat{\phi}(t)$). This suggests the mediation analysis over-estimates the total effect and, likely, the indirect effect. For political experience, the mediation estimates are systematically larger for candidates without political experience. This suggests the total effect found using the mediation analysis is an \emph{under-estimate}.

We can further apply the proposed sensitivity analysis from Section~\ref{sec:partial_test} and Appendix~\ref{app:sensitivity}, where the discrepancies between $\hat{\phi}(t)$ and $\hat{\alpha}(t,t)$ given $\bXi$ and $\bSi$, are used to roughly quantify the amount of bias that exists because of violations of the assumptions. Figure~\ref{fig:sens} shows the original results from Figure~\ref{fig:main_results} and those produced by the sensitivity analysis.

\begin{figure}[!htbp]
\caption{Sensitivity Analysis on Estimated Effects}\label{fig:sens} 
\includegraphics[width=\textwidth]{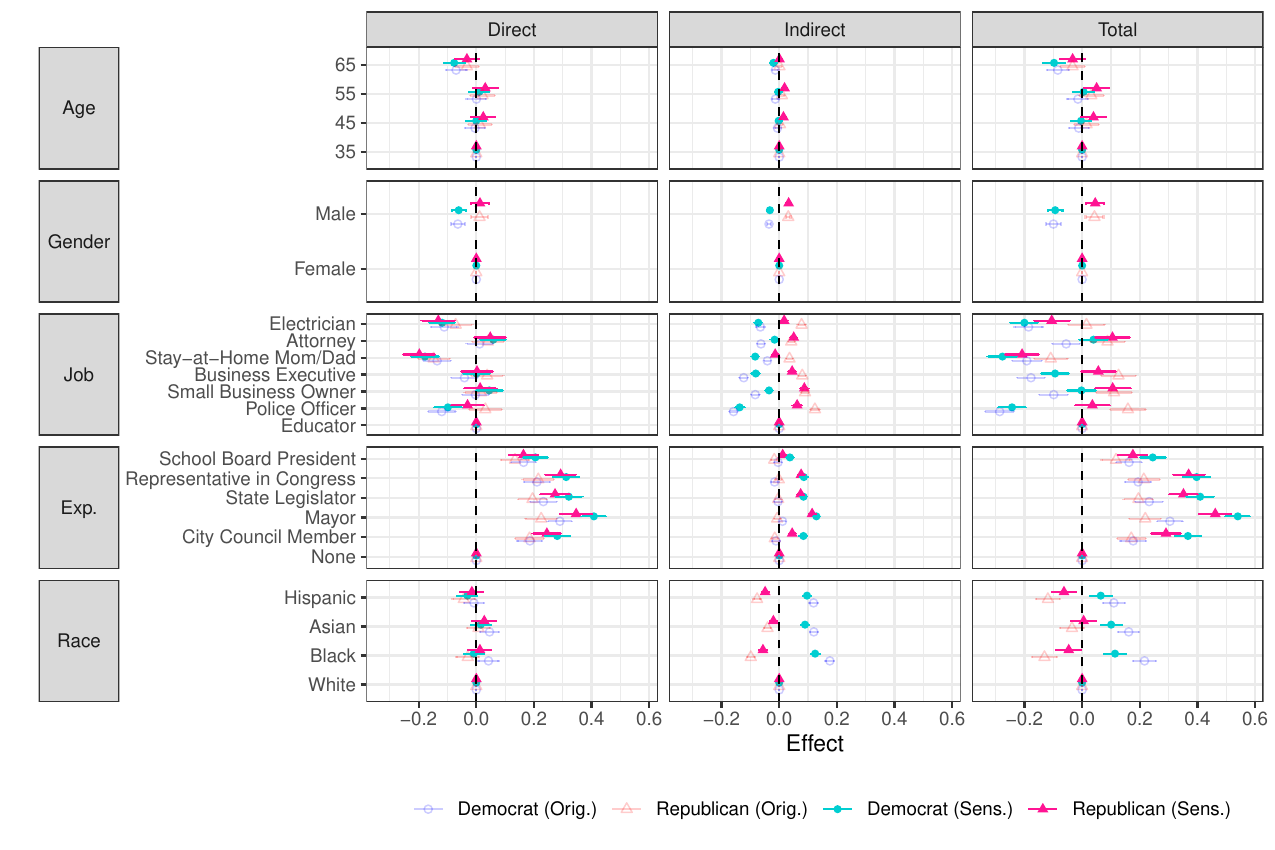}
\end{figure}

Broadly speaking, the results remain robust. The indirect effects for race and occupation are pulled towards zero after accounting for the sensitivity analysis but generally remain significant. Further, Democratic respondents still show both significant direct and indirect effects against police officers versus educators even after performing the sensitivity analysis. For political experience, the sensitivity analysis increases the estimated effects across all types (direct, indirect, and total) for both types of respondents, although the direct effect still remains considerably larger than the indirect effect.

\section{Conclusions}

Understanding why one factor in a conjoint experiment matters is essential to understanding causal mechanisms: Does it have a direct effect on the outcome or does its effect flow through its effects on another factor that, in turn, affects the outcome? Being clear about possible mediators that are randomized in the experiment allows for more precise interpretations of the causal estimands that are obtained by design from one's experiment. We show that, if one is willing to go beyond what is guaranteed by design alone, it is possible to decompose total effects into direct and indirect effects using the tools of causal mediation. Doing so can address important substantive questions such as understanding whether an effect is fully mediated through some other factor or recovering the total effect even when the mediator is randomized. 

In terms of experimental design, we suggest adding a simple additional experiment---asking respondents which mediator value they associate with the treatments---that, in conjunction with standard mediation assumptions allows for the identification of direct and indirect effects of some attribute in the conjoint experiment. Estimation can proceed using modern machine-learning techniques while having an estimator whose uncertainty can be easily quantified. Interpretable summaries of the underlying heterogeneous effects can be found using linear regression and presented alongside the average marginal direct and indirect effects. Further, some of the strong assumptions required by any mediation analysis can be partially tested if an auxiliary experiment is conducted. If no carry-over assumptions are credible in conjoint experiments across types of experiments, we also suggest other designs that could be implemented to relax these assumptions further. 

We applied our approach to an existing study about candidate choice; unlike other approaches that look at the eliminated effect between the total effect from a $Y(T)$-experiment and the controlled direct effect from a $Y(T,M)$-experiment, we find that the actual patterns of mediation are more complex but can be revealed with our experimental design (and corresponding assumptions). This suggests that the standard conjoint experiment often contains a large amount of theoretically interesting mediation and, with the design proposed in this paper, one can analyze that directly by a simple adjustment of the standard experimental protocol. Even if one does not wish to adopt the assumptions needed for causal mediation, analyzing the effects obtained from the $M(T)$-experiment provides important evidence about how treatment and a potential mediator are related.

\clearpage

\bibliography{paper_bib}
\bibliographystyle{apsr}

\clearpage

\appendix
\setcounter{page}{1}
\singlespacing

\section{Additional Illustration}\label{app:additional_illustration}

There are many scenarios in political science and marketing where the primary treatment variable affects beliefs about other characteristics (see, e.g., \citealt{Dafoe_Zhang_Caughey_2018,orr2020policy}). We focus on one example coming from a recent pair of studies published in the \emph{American Political Science Review} that use survey experiments to assess support for the democratic peace theory \citep{rathbun2025separate,tomz2026race}. The key stylized fact is that democracies rarely fight with each other. One explanation is that domestic publics are hostile to fight against other democracies because they are generally peaceful and because of a moral distaste for fighting against a democratic regime, a pattern supported by earlier survey experiments \citep{tomz2013public}. 

Recently, \citet{rathbun2025separate} put forward another explanation. When individuals are being told about a threat against another democracy, they really interpret it as a threat against a predominantly white country. In the framework of this paper, subjects are asked to report whether they would support military intervention after reading a vignette about a country who is close to developing nuclear weapons. The randomized treatment $T_i$ is whether the country is a democracy, the mediator $M_i$ is whether the country is predominately white, and $Y_i$ is whether the respondent would support American military intervention against the country's nuclear program. \cite{rathbun2025separate} perform both a $Y(T)$ and $Y(T,M)$-experiment and discuss various quantities including the eliminated effect, i.e., the difference between effects estimated in the $Y(T)$ and $Y(T,M)$-experiment. Their interpretation of the results is contested \citep{tomz2026race}; one area of disagreement is around the interpretation of the eliminated effect. As we also discuss in this paper, it is ambiguous whether an eliminated effect is associated with a non-zero indirect effect or is due to causal interaction. Our empirical example (Section~\ref{sec:empirical}) finds that it is highly variable whether an indirect effect accompanies an eliminated effect. Thus, while we do not wish to speak directly to this debate, it shows that estimating mediation quantities (e.g., indirect effects) is relevant to current empirical work; our paper provides an analytical and empirical framework (given the requisite assumptions) to do so.



\section{Proofs and Formalizations}

This section contains proofs or formalizations of claims made in the main text.

\subsection{Preliminary Assumptions}\label{app:SUTVA}

Section~\ref{sec:theory} notes that we make standard assumptions around SUTVA, consistency, and composition of the principal outcomes; these are standard in the mediation literature (e.g., \citealt{imai_2010_psychmeth,vanderweele_2014_epidem}), but we formalize them here in the binary treatment, binary mediator, and auxiliary treatment setting:
\begin{equation}
    \begin{split}
    Y_i &= \sum_{t, \bs} Y_i(t,\bs) \I{T_i = t, \bSi = \bs}; \quad M_i = \sum_{t,\bs} M_i(t,\bs) \I{T_i = t, \bSi = \bs} \\
    Y_i(0,\bs) &= Y_i(0, M_i(0),\bs); \quad Y_i(1,\bs) = Y_i(1, M_i(1),\bs).
    \end{split}
\end{equation}

\subsection{Interpreting Mediation Quantities}\label{app:example_direct}

The main text gives a thought experiment to understand the average marginal indirect effect $\delta(0)$. A similar example to understand the average marginal direct effect $\xi(0)$ is provided here. Imagine one fixed $\bSi = \bs$ and solicited the following information: (a) What party would you associate with a white candidate with other attributes $\bs$? The individual-level direct effect is obtained by comparing how the respondent would view a Black candidate ($T_i=1$) with the party provided in (a) versus a white candidate with the party provided in (a), i.e. $Y_i(1,M_i(0), \bs)-Y_i(0,M_i(0),\bs)$. Averaging over individuals and auxiliary covariates $\bs$ obtains the average marginal direct effect $\xi(0)$. This isolates the effect of race on candidate choice while holding constant the effect that race has on affecting beliefs about party.

In terms of controlled direct effects, the following illustration (directly adapted from \citealt{pear:01}) shows the difference between a controlled direct effect and a natural direct effect: Imagine there is a respondent who believes a candidate is a Democrat if and only if they are Black. Further, assume that they only wish to vote for Black candidates if they are Democrats.  A controlled direct effect exists: If we tell them the candidate is a Democrat, the treatment effect of race appears. However, a natural direct effect \emph{does not} exist: If we keep party fixed at the level that would be observed without treatment (i.e., Republican; the ``natural'' level), then there is no effect of providing the party label. 

\subsection{Alternative Mediation Quantities}\label{app:alt_mediation}

While it is common in social science to consider direct and indirect effects as discussed in the main text, \cite{vanderweele_2014_epidem} provides a different decomposition of the total effect into four parts: (a) controlled direct effect, (b) reference interaction, (c) mediated interaction, and (d) pure indirect effect. We restate his decomposition below, where $m^*$ is any value that $M$ may take on, adapting it for our notation and the idea of auxiliary treatments $\bs$. Note that this does not depend on any nested counterfactuals, although some terms can be re-expressed in that form.

\begin{equation}\label{eq:alt_mediation}
\begin{split}
    &Y_i(t,M_i(t),\bs) - Y_i(t',M_i(t'),\bs) = \underbrace{Y_i(t,m^*,\bs) -Y_i(t',m^*,\bs)}_{\textrm{Controlled Direct Effect}} + \\
    & \underbrace{\sum_{m} \left[Y_i(t,m,\bs)-Y_i(t',m,\bs) - Y_i(t,m^*,\bs) + Y_{i}(t',m^*,\bs)\right]\I{M_i(t',\bs)=m}}_{\textrm{Reference Interaction}} + \\
    &\underbrace{\sum_{m} \left[\begin{split}Y_i(t,m,\bs)-Y_i(t',m,\bs) \\ - Y_i(t,m^*,\bs) + Y_{i}(t',m^*,\bs)\end{split}\right]\left[\I{M_i(t,\bs)=m}-\I{M_i(t',\bs)=m}\right]}_{\textrm{Mediated Interaction}}+ \\
    &\underbrace{\sum_m\left[Y_i(t',m,\bs)-Y_i(t',m^*,\bs)\right]\left[\I{M_i(t,\bs)=m}-\I{M_i(t',\bs)=m}\right]}_{\textrm{Pure Indirect Effect}}
\end{split}
\end{equation}

The proof follows immediately from \citet[p. 759]{vanderweele_2014_epidem} where one notes that $(t,\bs)$ and $(t',\bs)$ defines a contrast of (high-dimensional) treatments. Marginalizing over units $i$ and auxiliary treatments $\bs$ gives a decomposition of the average marginal total effect into an AMCE (i.e., an average marginal controlled direct effects), an average marginal reference interaction, an average marginal mediated interaction, and an average marginal pure indirect effect. As follows immediately from \cite{vanderweele_2014_epidem}, the average marginal direct effect in the main text is the sum of the average marginal controlled direct effect and the average marginal reference interaction. The average marginal indirect effect is the sum of the latter two terms---the average marginal mediated interaction and the average marginal pure indirect effect. 

Under our principal ignorability assumptions, each of the four decomposed parts can be identified from the data and estimated in a doubly robust fashion. We focus primarily on the standard two-way decomposition and thus do not explicitly derive this decomposition, but \citet[p. 760]{vanderweele_2014_epidem} provides the identifiable estimand and analogous application of \cite{kennedy2024semiparametric}'s results would produce an influence function.

\subsection{Effect of 
``No Interactions''}\label{app:no_interaction}

As the text notes, without an $M(T)$-study, even principal ignorability (or sequential ignorability) is insufficient to identify a direct and indirect effect. This is because no information is learned about the treatment-mediator mapping. However, if one makes very strong assumptions, it is possible to learn the direct and indirect effect.

We illustrate this using the ``no interaction'' assumption from \cite{imai_2013_jrss}. In our notation, this states that, for every $i$ and every $\bs$,

\begin{equation}
    Y_i(t,m,\bs)-Y_i(t',m,\bs) = Y_i(t,m',\bs) - Y_i(t', m',\bs).
\end{equation}

If one makes this assumption, the average marginal indirect effect can be estimated as the difference between the average marginal total effect (i.e., the average marginal component effect; see Table~\ref{tab:QOI}) from the $Y(T)$-experiment and the controlled direct effect or average marginal component effect on $T$ from the $Y(T,M)$-experiment \citep{imai_2013_jrss,acharya_2018_polan}. The cost of doing so is assuming that the controlled direct effect equals the direct effect, i.e., assuming away the existence of a reference interaction \citep{vanderweele_2014_epidem}. In the notation of Appendix~\ref{app:alt_mediation}, above, the no interaction assumption eliminates the reference interaction from Equation~\ref{eq:alt_mediation} and thus the average marginal indirect effect is equivalent to the (average marginal) eliminated effect.

\subsection{Proof of Equation~\ref{eq:mediation_fmla}}\label{app:mediation_fmla}

We prove the mediation formula using Assumptions~\ref{eq:A_PI}-\ref{eq:A_RandMT}. This derivation holds for multi-valued treatments and mediators. We assume that for some distribution or weights on $\bSi$, call it $f(\bs)$, which we take to be independent of $t$, $M$, and $\bm{X}$, the goal is to obtain the (marginalized) average potential outcome

\begin{equation}
    \alpha_f(t,t') = \E[\E_{\bm{S}\sim f}[Y_i(t,M_i(t',\bm{S}),\bm{S})]] = \sum_{\bs} \E[Y_i(t, M_i(t',\bs),\bs)] f(\bm{S}=\bs)
\end{equation}

This can be expressed as a combination of identifiable pieces as follows:

\begin{equation}\label{eq:proof_id}
\begin{split}
\alpha_f(t,t') = &\sum_{\bs, \bx, \bm{m}} \left[\begin{split}&\E[Y_i(t, m,\bs) | \I{M_i(t', \bs) = m}, \bXi=\bx] \times  \\ &\mathrm{Pr}\left(M_i(t', \bs) = m | \bXi = \bx\right) \times\mathrm{Pr}\left(\bXi = \bx\right)f(\bm{S} = \bs)\end{split}\right] \\
&\sum_{\bs, \bx, \bm{m}} \left[\begin{split}&\E[Y_i(t, m,\bs) | \bXi]\times   \\ &\mathrm{Pr}\left(M_i(t', \bs) = m | \bXi = \bx \right) \times\mathrm{Pr}\left(\bXi = \bx\right) f\left(\bm{S} =\bs\right)\end{split}\right] \\
&\sum_{\bs, \bx, \bm{m}} \left[\begin{split}&\E[Y_i(t, m, \bs) | \bXi,A_i = 1] \times \\ &\mathrm{Pr}\left(M_i(t',\bs) = m | \bXi = \bx, A_i = 0 \right) \times \mathrm{Pr}\left(\bXi = \bx\right)\mathrm{Pr}(\bm{S}=\bs)\end{split}\right] \\
&\sum_{\bs, \bx, \bm{m}} \left[\begin{split}&\E[Y_i(t, m, \bs) | T_i = t, M_i = m, \bXi, \bSi = \bs, A_i = 1] \times\\& \mathrm{Pr}\left(M_i(t',\bs) = m | T_i = t', \bXi = \bx, \bSi = \bs, A_i = 0 \right) \times \mathrm{Pr}\left(\bXi = \bx\right)f(\bm{S}=\bs)\end{split}\right] \\
&\sum_{\bs, \bx, \bm{m}} \left[\begin{split}&\E[Y_i | T_i = t, M_i = m, \bXi, \bSi = \bs, A_i = 1] \times 
\\&\mathrm{Pr}\left(M_i = m | T_i = t', \bXi = \bx, \bSi = \bs, A_i = 0 \right) \times \mathrm{Pr}\left(\bXi = \bx\right)f(\bm{S}=\bs)
\end{split}\right]
\end{split}
\end{equation}

The first line follows by the law of iterated expectations. The second line follows principal ignorability (Assumption~\ref{eq:A_PI}). The third follows given the ignorability of the experiment (Assumption~\ref{eq:A_RandExp}) and, implicitly, by the manipulation exclusion restriction (Assumption~\ref{eq:A_ExcluRes}). The fourth follows given the ignorability of the treatment(s) and mediator (given it is randomly assigned) in the $Y(T,M)$ (\ref{eq:A_RandYTM}) and the ignorability of the treatment(s) in the $M(T)$-experiment (Assumption~\ref{eq:A_RandMT}). The fifth line follows by consistency (see Appendix~\ref{app:SUTVA}), i.e., where the observed $Y_i$ value in the $Y(T,M)$ experiment is the potential outcome corresponding to the observed $T_i$, $M_i$, and $\bSi$; a similar logic applies for the observed $M_i$ in the $M(T)$-experiment. Different choices of $f(\bm{S} = \bs)$ would allow for weighting the average potential outcomes to target different quantities of interest as in \cite{de2022improving}.

In practice, we use the empirical distribution of $\bSi$ for $f(\bm{S})$, denoting this as $\alpha(t,t')$ to match the main text, the following simplification holds noting that $\bXi$ and $\bSi$ are independent as $\bSi$ is randomly assigned.

\begin{equation}
\alpha(t,t') = \sum_{\bs, \bx, \bm{m}} \left[\begin{aligned}&\E[Y_i | T_i = t, M_i = m, \bXi, \bSi = \bs, A_i = 1] \times 
\\&\mathrm{Pr}\left(M_i = m | T_i = t', \bXi = \bx, \bSi = \bs, A_i = 0 \right) \times \mathrm{Pr}\left(\bXi = \bx, \bSi = \bs\right)
\end{aligned}\right]
\end{equation}

\subsection{Principal Strata and Direct and Indirect Effects}\label{app:decompose_PI}

This section provides an interpretation of the average marginal direct and indirect effects in terms of principal strata. We can define $\mathcal{G}_1$ as the principal strata where $M_i(1,\bs) = 1$ and $M_i(0,\bs) = 0$ and $\mathcal{G}_2$ as those where $M_i(1,\bs) = 0$ and $M_i(0,\bs) = 1$. These are known as associative strata \citep{fran:rubi:02,forastiere_2018_biomet}. Equation~\ref{eq:decompose_PI} decomposes the average marginal indirect effect $\delta(t)$ as follows

\begin{equation}\label{eq:decompose_PI}
\begin{split}
 \delta(t) &= \E[Y_i(t,M_i(1),\bSi)-Y_i(t,M_i(0),\bSi)] \\
 &= \sum_{g} \E[Y_i(t,M_i(1),\bSi)-Y_i(t,M_i(0),\bSi) | G_i = g] \times \mathrm{Pr}(G_i = g) \\  &= \left\{\begin{split}&\sum_{g \in \mathcal{G}_1} \E[Y_i(t,1,\bSi)-Y_i(t,0,\bSi) | G_i = g] \times \mathrm{Pr}(G_i = g) + \\ &\sum_{g \in \mathcal{G}_2} \E[Y_i(t,0,\bSi)-Y_i(t,1,\bSi) | G_i = g] \times \mathrm{Pr}(G_i = g)  \end{split}\right\}.
\end{split}  
\end{equation}

The final line shows only those strata in $\mathcal{G}_1$ and $\mathcal{G}_2$ contribute to $\delta(t)$, i.e. by their size $\mathrm{Pr}(G_i = g)$. Recalling that $\E[Y_i(t,m,\bSi)-Y_i(t,m',\bSi)|G_i = g]$ is an AMCE for $m$ vs $m'$, the average marginal indirect effect is a weighted combination of principal strata-specific AMCEs. Equation~\ref{eq:decompose_PI_DE} treats the average marginal direct effect.

\begin{equation}\label{eq:decompose_PI_DE}
\begin{split}
 \xi(t) &= \sum_g \E[Y_i(1,M_i(t),\bSi)-Y_i(0,M_i(t),\bSi) | G_i] \mathrm{Pr}(G_i = g)
\end{split}  
\end{equation}

This depends on all strata, unlike the average marginal indirect effect, but uses a strata-specific quantity---a combination of AMCE and average combined effects [ACE], i.e., in the strata where $M_i(t,\bs) \neq M_i(t', \bs)$. Note that the average marginal indirect effect considers the \emph{same} AMCE within $\mathcal{G}_1$ and $\mathcal{G}_2$, respectively.

\subsection{Influence Function}\label{app:influ}

This sub-section derives the influence function for $\alpha(t,t')$, focusing first on the version defined in terms of the empirical distribution. Recall that

\begin{equation}
\alpha(t,t') =\sum_{\bs, \bx, \bm{m}} \left[\begin{aligned}&\E[Y_i | T_i = t, M_i = m, \bXi = \bx, \bSi = \bs, A_i = 1] \times 
\\&\mathrm{Pr}\left(M_i = m | T_i = t', \bXi = \bx, \bSi = \bs, A_i = 0 \right) \times \mathrm{Pr}\left(\bXi = \bx,\bSi=\bs\right)
\end{aligned}\right].
\end{equation}

Using the tools in \cite{kennedy2024semiparametric}, the influence function can be expressed as

\begin{equation}\label{eq:if_empirical}
    \mathrm{IF}\left(\alpha(t,t')\right) =  \begin{aligned}
    &\frac{\I{A_i=1,T_i=t} \cdot e_{M_i}(t', x, \bSi)}{p^T_t(\bXi, \bSi, 1) \cdot p^M_{M_i}(t, \bXi, \bSi) \cdot p_a(\bXi,\bSi)} \left(Y_i - \mu_Y(t, M_i, \bXi, \bSi)\right) + \\
     &\frac{\I{T_i = t', A_i = 0}}{p^T_{t'}(\bXi,\bSi, 0)(1-p_a(\bXi, \bSi))} \left(\begin{aligned}&\mu_Y(t,M_i,\bXi,\bSi) \\ &- \sum_m \mu_Y(t,m,\bXi, \bSi)e_m(t', \bXi, \bSi)\end{aligned}\right) +  \\
     &\sum_{m} \mu_Y(t,m,\bXi,\bSi)e_m(t', \bXi, \bSi) 
    \\&- \alpha(t,t'). \end{aligned}
\end{equation} 

The proof follows by the results in \cite{kennedy2024semiparametric}, specifically that
\begin{equation}
    \begin{aligned}
    &\mathrm{IF}\left(\mu_Y(t, m, \bx, \bs)\right)= \frac{\I{T_i = t, M_i = m, \bXi = \bx, \bSi = \bs, A_i = 1}}{\mathrm{Pr}(T_i = t, M_i = m, \bXi = \bx, \bSi = \bs, A_i = 1)} \left(Y_i - \mu_Y(t, m, \bx, \bs)\right), \\
    &\mathrm{IF}\left(e_m(t, \bx, \bs)\right)= \frac{\I{T_i = t, \bXi = \bx, \bSi = \bs, A_i = 0}}{\mathrm{Pr}(T_i = t, \bXi = \bx, \bSi = \bs, A_i = 0)} \left(\I{M_i = m} - e_m(t, \bx, \bs)\right), \\
    &\mathrm{IF}\left(p(\bx,\bs)\right) = \I{\bXi = \bx, \bSi = \bs} - p(\bx,\bs), \\
    & \mathrm{IF}\left(a b\right) = \mathrm{IF}\left(a\right) \times b + a \times \mathrm{IF}\left(b\right).
    \end{aligned}
\end{equation}

We can use our design to simplify this further; for example, $p_a(\bx,\bs) = p_a(\bx)$ because $\bSi$ is randomized, $p^T_t(\bx,\bs,1) = p^T_t(\bs,1)$ as $(T_i, \bSi)$ are randomized, and $p^M_m(t, \bx, \bs)$ simplifies to $p^M_m(t, \bs)$ as $(M_i, T_i, \bSi)$ are randomized when $A_i = 1$. Our proposed influence function thus nearly exactly coincides with \cite{farbmacher_2022_econmet} once we adjust for the propensity score of being in each experiment, i.e., $A_i \in \{0,1\}$ and account for the missing data in each experiment recall that if $A_i = 0$, $Y_i$ is missing and if $A_i = 1$, then $M_i(T_i)$ is missing. Equation~\ref{eq:if_empirical} also accounts for the fact that $M_i$ is fixed (versus set at its natural value) if $A_i = 1$. Given this influence function, the estimator proposed in the main text follows as either a one-step estimator or by solving the implied estimation equation for $\alpha(t,t')$ \citep{kennedy2024semiparametric}.

We also consider the case where some other distribution $f$  over $\bSi$ is used. Given that $f$ is assumed to be independent of the distribution of $\bm{X}$, this example mirrors the case of stochastic intervention \citep{kennedy2024semiparametric}. The influence function for $\alpha_f(t,t')$ is shown below, where $\pi(\bm{S})$ denotes the randomization distribution of $\bSi$ fixed in the experiment.

\begin{equation}\label{eq:if_marginalized}
\mathrm{IF}\left(\alpha_f(t,t')\right) = \begin{aligned}
    &\frac{\I{A_i=1,T_i=t} \cdot e_{M_i}(t', x, \bSi) f(\bSi)}{p^T_t(\bXi, \bSi, 1) \cdot p^M_{M_i}(t, \bXi, \bSi) \cdot p_a(\bXi,\bSi)\pi(\bSi)} \left(Y_i - \mu_Y(t, M_i, \bXi, \bSi)\right) +\\
     &\frac{\I{T_i = t', A_i = 0} f(\bSi)}{p^T_{t'}(\bXi,\bSi, 0)(1-p_a(\bXi, \bSi))\pi(\bSi)} \left(\begin{aligned}&\mu_Y(t,M_i,\bXi,\bSi) \\&-\sum_m \mu_Y(t,m,\bXi, \bSi)e_m(t', \bXi, \bSi)\end{aligned}\right) + \\
     & \E_{\bm{S} \sim f}\left[\sum_{m} \mu_Y(t,m,\bXi,\bm{S})e_m(t', \bXi, \bm{S})\right]\\
     &- \alpha_f(t,t')
    \end{aligned}
\end{equation} 

This aligns nearly exactly with the empirical distribution influence function with two critical differences; first, note that there is an adjustment in the doubly robust terms $f(\bSi)/\pi(\bSi)$ to re-weight those terms to match $f$---instead of $\pi$. Second, note that the third term now involves an expectation over $\bm{S}$ versus merely plugging in the observed value. This is, indeed, what would be required to estimate the $\alpha_f(t,t')$ even if a non-robust method was used. In practice, $\bm{S}$ is very high-dimensional so this integral is infeasible to do exactly. Thus, we suggest doing Monte Carlo integration where some number of samples are taken from $f$ and then used to approximate this third term. Note that as we are averaging this across all observations in the estimation stage, i.e., empirically averaging across $\bXi$, the number of samples likely need not be so large to obtain a stable estimate.

It is useful to compare $\alpha(t,t')$ and $\alpha_\pi(t,t')$, i.e., using the empirical distribution versus the second approach where $f = \pi$ (i.e., the randomization distribution). In this case, $f(\bs)/\pi(\bs) = 1$ and thus the only difference between this influence function and the one above is that the third term is computed by averaging over $\bs$ versus using the observed value. As $\bSi$ is randomly generated, however, this is equivalent to doing Monte Carlo integration with a single sample. 

Put another way, the influence function that we use with the empirical distribution (Equation~\ref{eq:if_empirical}) is derived assuming the distribution of $\mathrm{Pr}(\bXi=\bm{x},\bSi=\bm{s})$ is unknown. If we use the fact that $\bXi$ and $\bSi$ are independent and the distribution of $\bSi$ is known, we can obtain the second influence function (Equation~\ref{eq:if_marginalized}) that is likely more efficient but requires Monte Carlo integration to deal with the third term. For simplicity, we thus prefer the first option in our analyses.

\subsection{Crossover Design}\label{app:crossover}

Principal ignorability (Assumption~\ref{eq:A_PI}) is often viewed as a strong assumption. We note that this could be avoided using a crossover design, adapting the one proposed in \cite{imai_2013_jrss} to conjoint experiments---if one is willing to assume consistency and no carry-over effects. While it is outside of the scope of this paper to fully pursue this, one might mimic an idealized experiment (see Section~\ref{sec:theory} and Appendix~\ref{app:example_direct}) where one carefully designs the repeated tasks a respondent performs. If one assumes that repeatedly performing tasks from the $M(T)$ and $Y(T,M)$-experiments in close proximity does not change how they respond (i.e., a strong ``no carry-over'' assumption), then this can identify the relevant indirect and direct effects without requiring the assumption of principal ignorability. Various things could make this more plausible, e.g., re-interviewing subjects weeks apart, although those often raise the cost of the experiment. 

If one desires to do this in a single study, the assumption may also be rather strong. However, recent work on conjoint experiments suggests that they are sufficiently cognitively demanding that asking respondents to, say, rate the same pair of profiles \emph{again}---in rather quick succession after they initially did so---does not seem to elicit any negative feedback \citep{clayton2025correcting}. Thus, an interesting area for future research would be to extend our approach to use a no carry-over design, although it would be important to have a way to validate the reasonableness of the no carry-over assumption---perhaps extending ideas from \citet[p. 17]{imai_2013_jrss}.

\subsection{Multiple Mediators}\label{app:multiple_mediators}

Assume that there were two mediators of interest, say party ($M_1$; Democrat, Republican, and Independent) and policy ($M_2$; the federal government should make it [easier / harder] to purchase guns; e.g., \citealt{orr2020policy}). Formally treating the case of multiple mediators is challenging and outside of scope for this paper. Assuming both mediators are low-dimensional and categorical, one can create a single combined mediator (i.e., of all unique combinations of $M_1$ and $M_2$) and use our framework. It would require the following (minimal) changes to the proposed design.

\begin{itemize}
    \item $Y(T,M)$-Experiment: Randomize both $M_1$ and $M_2$ and $(M_1, M_2)$ to respondents.
    \item $M(T)$-Experiment: Solicit information on how the respondents think that $T$ (and $\bm{S}$) map onto $M_1$ and $M_2$. This can be done in various ways (e.g., asking separate questions for $M_1$ and $M_2$), but it is essential that the \emph{joint} $(M_1, M_2)$ information is elicited each respondent in the $M(T)$-Experiment.
    \item Analysis: Our method and estimation using influence functions does not restrict the number of (discrete) levels the mediator can have. Thus, the only change needed is to estimate $\mathrm{Pr}(M_{1,i} =m _1,M_{2,i} = m_2 | T_i, \bSi, \bXi)$ from the $M(T)$-experiment, i.e., the \emph{joint} distribution of the multiple mediators given $T_i, \bSi, \bXi$.    
\end{itemize}

\section{Estimation Using Linear Regression}\label{app:regestimate}

This appendix shows linear regression with appropriately clustered standard errors can estimate the average marginal direct, indirect, and total effects. All of our primary quantities of interest can be expressed as $F(\bm{z}) = \E[\bm{z}^T\bm{\psi}_i]$ where $\bm{\psi}_i$ is a $|\mathcal{T}|\times |\mathcal{T}|$ length vector collecting all $\{\psi_i(t,t')\}$. For example, $\alpha(t,t') = F(\bm{z}_1)$ where $\bm{z}_1$ has a `1' in the position of $(t,t')$ and zero otherwise. $\xi(t^*) = \alpha(t,t^*) - \alpha(t',t^*)$ can be found by a $\bm{z}$ with `1' in the $(t,t^*)$-position and `-1' in the $(t',t^*)$-position. $\bar{\xi}$ can be found similarly, e.g., $1/|\mathcal{T}|$ in the $(t,t^*)$-position for all $t^* \in \mathcal{T}$ and $-1/|\mathcal{T}|$ in the $(t',t^*)$-position for all $t^* \in \mathcal{T}$.

Thus, this section considers estimating $F(\bm{z})$ for an arbitrary $\bm{z}$. To do this, we define our estimator, noting that it is consistent and asymptotically normal as it is a linear combination of influence functions and cross-fitting to obtain $\hat{\bm{\psi}}_i$ \citep{farbmacher_2022_econmet,kennedy2024semiparametric}.

\begin{equation}\label{eq:fz}
\begin{split}
    &\hat{F}(\bm{z}) = \frac{1}{N}\sum_{i=1}^N \bm{z}^T\hat{\bm{\psi}}_i; \quad \hat{\sigma}^2(\bm{z}) = \frac{1}{N} \sum_{i=1}^N \left(\bm{z}^T\hat{\bm{\psi}}_i - \hat{F}(\bm{z})\right)^2 \\
    &\frac{\sqrt{N}\left(\hat{F}(\bm{z}) - F(\bm{z})\right)}{\sqrt{\hat{\sigma}^2(\bm{z})}} \to^d N\left(0,1\right)
\end{split}
\end{equation}

This exactly matches our definition of $\hat{\alpha}(t,t')$ and $\hat{\delta}(t)$ in the main text for specific choices of $\bm{z}$. Theorem~\ref{thm:estimation} provides three equivalent ways of obtaining $\hat{F}(\bm{z})$ and the accompanying standard error. The theorem considers a more general case where a linear regression with covariates $\bm{w}_i$ are used to predict $\bm{z}^T\bm{\psi}_i$ and obtain the usual consistent and asymptotically normal estimator of the population linear regression coefficient $\bm{\beta} =\mathbb{E}[\bm{w}_i \bm{w}_i^T]^{-1} \mathbb{E}[\bm{w}_i\left(\bm{z}^T\bm{\psi}_i\right)]$. This handles heterogeneous effect estimation in the following section; $\bm{w}_i = 1$ recovers the average effect in Equation~\ref{eq:fz}.

\begin{theorem}\label{thm:estimation}
    Assume a full-rank design matrix $\bm{W} \in \mathbb{R}^{N \times p}$ and a matrix of outcomes $\bm{Y} \in \mathbb{R}^{N \times q}$. For any vector $\bm{z}$, consider three estimators of the population linear regression coefficient $\bm{\beta}$ with outcome $\bm{y}_i^T\bm{z}$ and regressors $\bm{w}_i$:

    \begin{itemize}
        \item Approach 1: Define $\tilde{\bm{y}} = \bm{Y}\bm{z}$ and $\hat{\bm{\beta}}_1 = \left(\bm{W}^T\bm{W}\right)^{-1}\bm{W}^T\tilde{\bm{y}}$ with a HC0-robust covariance matrix, i.e.
        $$\widehat{\mathrm{Var}}(\hat{\bm{\beta}}_1) = \left(\frac{1}{N} ~\bm{W}^T\bm{W}\right)^{-1}\left(\frac{1}{N}\sum_{i=1}^N \left(\tilde{y}_i - \bm{z}_i^T\hat{\bm{\beta}}_1\right)^2 \bm{z}_i\bm{z}_i^T\right)\left(\frac{1}{N} ~\bm{W}^T\bm{W}\right)^{-1}.$$
        \item Approach 2: Perform $q$ separate regressions for each column of $\bm{Y}$, i.e. $\hat{\bm{\beta}}_q = \left(\bm{W}^T\bm{W}\right)^{-1}\bm{W}^T\bm{y}_q$. The (co)-variance matrix $\widehat{\mathrm{Cov}}\left(\hat{\bm{\beta}}_q, \hat{\bm{\beta}}_{q'}\right)$ is defined as
        $$\left(\frac{1}{N}~ \bm{W}^T\bm{W}\right)^{-1} \left(\frac{1}{N}\sum_{i=1}^N \left(y_{i,q} - \bm{w}_i^T\hat{\bm{\beta}}_q\right)\left(y_{i,q'}-\bm{w}_i^T\hat{\bm{\beta}}_{q'}\right)\bm{w}_i\bm{w}_i^T\right)\left(\frac{1}{N}~\bm{W}^T\bm{W}\right)^{-1}$$

        Then, define $\hat{\bm{B}} = \left[\hat{\bm{\beta}}_1|\hat{\bm{\beta}}_2 | \cdots |\hat{\bm{\beta}}_q\right]$ and $\hat{\bm{\beta}}_2 = \hat{\bm{B}}\bm{z} = \sum_{j=1}^q \hat{\bm{\beta}}_j z_j$.
        \item Approach 3: Perform a single linear regression, vectorizing $\bm{Y}$ column-wise, i.e. $\check{\bm{y}}^T = \left[\bm{y}_1^T, \bm{y}_2^T, \cdots, \bm{y}_q^T\right]^T$. Define the corresponding design matrix $\check{\bm{W}} = \bm{I} \otimes \bm{W}$, i.e. a block diagonal matrix with $q$ blocks each consisting of $\bm{W}$. Define $\hat{\check{\bm{\beta}}} = \left(\check{\bm{W}}^T \check{\bm{W}}\right)^{-1}\check{\bm{W}}^T \check{\bm{y}}$. The estimate of the variance matrix uses cluster-robust HC0 errors where the errors are clustered by observation $i$. Formally,

        $$\widehat{\mathrm{Var}}(\hat{\check{\bm{\beta}}})=\left(\frac{1}{N} \check{\bm{W}}^T\check{\bm{W}}\right)^{-1}\left(\frac{1}{N}\sum_{i=1}^N \left(\bm{I}_q \otimes \bm{w}_i\right)\hat{\bm{u}}_i  \hat{\bm{u}}_i^T \left(\bm{I}_q \otimes \bm{w}_i^T\right)\right)\left(\frac{1}{N} \check{\bm{W}}^T\check{\bm{W}}\right)^{-1},$$

        where $\bm{y}_i$ are the $q$ observations associated with unit $i$ and the vector of residuals $\hat{\bm{u}}_i = \bm{y}_i -\left(\bm{I}_q \otimes \bm{w}_i^T\right)\hat{\check{\bm{\beta}}}$. Define $\hat{\bm{\beta}}_3 = \left(\bm{z}^T \otimes \bm{I}_p\right) \hat{\check{\bm{\beta}}}$.
        
    \end{itemize}

    All three approaches are equivalent, i.e. $\hat{\bm{\beta}}_1 = \hat{\bm{\beta}}_2 = \hat{\bm{\beta}}_3$ as well as $\widehat{\mathrm{Var}}(\bm{\beta}_1) = \widehat{\mathrm{Var}}(\bm{\beta}_2) = \widehat{\mathrm{Var}}(\bm{\beta}_3)$.
\end{theorem}

\emph{Proof:} Consider first $\hat{\bm{\beta}}_2$ and its variance matrix. $\hat{\bm{\beta}}_1 = \hat{\bm{\beta}}_2$ as $\hat{\bm{B}}=\left(\bm{W}^T\bm{W}\right)^{-1}\bm{W}^T\bm{Y}$ and $\hat{\bm{\beta}}_2 = \hat{\bm{B}}\bm{z} = \left(\bm{W}^T\bm{W}\right)^{-1}\bm{W}^T\bm{Y}\bm{z} = \left(\bm{W}^T\bm{W}\right)^{-1}\bm{W}^T\tilde{\bm{y}} =\hat{\bm{\beta}}_1$. The equivalence of the estimated variance matrix is shown below,

\begin{align*}
\widehat{\mathrm{Var}}\left(\hat{\bm{\beta}}_2\right) &= \sum_{j=1}^q \sum_{j'=1}^q z_j \cdot z_{j'} \cdot \widehat{\mathrm{Cov}}\left(\hat{\bm{\beta}}_q, \hat{\bm{\beta}}_{q'}\right) \\
&=\left(\frac{1}{N}~ \bm{W}^T\bm{W}\right)^{-1} \left(\frac{1}{N}\sum_{i=1}^N \left[\bm{z}^T \left(\bm{y}_i - \hat{\bm{B}}^T\bx_i\right)\left(\bm{y}_i -\hat{\bm{B}}^T \bx_i\right)^T\bm{z}\right]\bm{w}_i\bm{w}_i^T\right)\left(\frac{1}{N}~\bm{W}^T\bm{W}\right)^{-1}\\
&=\left(\frac{1}{N}~ \bm{W}^T\bm{W}\right)^{-1} \left(\frac{1}{N}\sum_{i=1}^N \left(\tilde{\bm{y}}_i - \hat{\bm{\beta}}_1^T \bm{w}_i\right)^2\bm{w}_i\bm{w}_i^T\right)\left(\frac{1}{N}~\bm{W}^T\bm{W}\right)^{-1}.
\end{align*}

Next, consider $\hat{\bm{\beta}}_3$. Its equivalence to $\hat{\bm{\beta}}_1$ is shown below, noting that $\mathrm{vec}^{-1}(\check{\bm{y}}) = \bm{Y}$ by construction and $\mathrm{vec}(\bm{A}\bm{B}\bm{C}) = \left(\bm{C}^T \otimes \bm{A}\right)\mathrm{vec}(\bm{B})$:
\begin{align*}
\hat{\bm{\beta}}_3 &= \left(\bm{z}^T \otimes \bm{I}_p\right) \left(\bm{I} \otimes \left[\bm{W}^T\bm{W}\right]^{-1}\bm{W}^T\right)\check{\bm{y}} \\
&=\left(\bm{z}^T \otimes \left[\bm{W}^T\bm{W}\right]^{-1}\bm{W}^T\right)\check{\bm{y}} \\
&= \mathrm{vec}\left(\left[\bm{W}^T\bm{W}\right]^{-1}\bm{W}^T\bm{Y}\bm{z} \right) = \hat{\bm{\beta}}_1.
\end{align*}

Regarding the estimated variance matrix, we first note the meat can be simplified using the mixed product properties of Kronecker product, noting that for a vector $\bm{a}$, $\bm{a} = \bm{a} \otimes 1$:

$$\frac{1}{N}\sum_{i=1}^N \left(\bm{I}_q \otimes \bm{w}_i\right)  \hat{\bm{u}}_i \hat{\bm{u}}_i^T \left(\bm{I}_q \otimes \bm{w}_i^T\right) = \frac{1}{N}\sum_{i=1}^N \left(\hat{\bm{u}}_i \otimes \bm{w}_i\right)\left(\hat{\bm{u}}_i \otimes \bm{w}_i\right)^T.$$

Next, recall that $\check{\bm{W}}^T\check{\bm{W}} = \bm{I}_q \otimes \bm{W}^T\bm{W}$ and define $\bm{A} = \left(\frac{1}{N}\left[\bm{W}^T\bm{W}\right]\right)^{-1}$. With some re-arrangement
\begin{align*}
    \widehat{\mathrm{Cov}}(\hat{\bm{\beta}}_3) &=\left(\bm{z}^T \otimes \bm{A}\right) \left(\frac{1}{N}\sum_{i=1}^N \left(\hat{\bm{u}}_i \otimes \bm{w}_i\right)\left(\hat{\bm{u}}_i \otimes \bm{w}_i\right)^T\right)\left(\bm{z} \otimes \bm{A}\right) \\
    &= 
    \frac{1}{N}\sum_{i=1}^N \left(\bm{z}^T\hat{\bm{u}}_i \otimes \bm{A}\bm{w}_i\right)\left(\bm{z}^T\hat{\bm{u}}_i \otimes \bm{A}\bm{w}_i\right)^T \\
    &=\bm{A}\left(\frac{1}{N}\sum_{i=1}^N \left(\bm{z}^T\hat{\bm{u}}_i\right)^2 \bm{w}_i\bm{w}_i^T\right)\bm{A} \\&=\left(\frac{1}{N}\bm{W}^T\bm{W}\right)^{-1}\left(\frac{1}{N}\sum_{i=1}^N \left(\tilde{\bm{y}}_i - \hat{\bm{\beta}}_1^T\bm{w}_i\right)^2 \bm{w}_i\bm{w}_i^T\right)\left(\frac{1}{N}\bm{W}^T\bm{W}\right)^{-1}.
\end{align*}
The last line follows by noting that $\bm{z}^T\hat{\bm{u}}_i = \tilde{\bm{y}}_i - \bm{z}^T\left(\bm{I} \otimes \bm{w}_i^T\right)\hat{\check{\bm{\beta}}}$, $\left(\bm{I} \otimes \bm{w}_i^T\right)\hat{\check{\bm{\beta}}} = \check{\bm{B}}^T \bm{w}_i$ by simple re-arrangement, where $\check{\bm{B}} = \mathrm{vec}^{-1}\left(\hat{\check{\bm{\beta}}}\right)$ and $\hat{\bm{\beta}}_3 = \check{\bm{B}} \bm{z} b$. Thus, $\bm{z}^T\check{\bm{B}}^T\bm{w}_i = \hat{\bm{\beta}}_1^T\bm{w}_i$ by the equivalence of $\hat{\bm{\beta}}_3 = \hat{\bm{\beta}}_1$. \qed

Theorem~\ref{thm:estimation} is useful in noting the equivalence of the following approaches: (i) estimating a regression on the transformed outcome $\bm{z}^T\bm{y}_i$ directly, (ii) using separate regressions and then combining the $\hat{\bm{\beta}}_q$ directly, or (iii) using a single regression where, in effect, the covariates are interacted with indicator variables for the column $q$ from which the vectorized data corresponds. In practice, if one wishes to use many different $\bm{z}$, approaches (ii) and (iii) are superior as the main cost of fitting the regression is performed only once. Further, note that if some elements of $\bm{z}$ are zero, the corresponding columns of $\bm{Y}$ can be dropped from the analysis and all three approaches remain equivalent.

Given Theorem~\ref{thm:estimation}, estimating $\hat{F}(\bm{z})$ is a specific case; Corollary~\ref{coro:est_avg} states this formally:

\begin{coro}\label{coro:est_avg}
$\hat{F}(\bm{z})$ (Equation~\ref{eq:fz}) can be estimated using any of the three approaches in Theorem~\ref{thm:estimation}, where $\bm{W} = \bm{1}$, i.e. an intercept-only regression.    
\end{coro}

\emph{Proof:} Direction estimation using the influence functions (i.e., Equation~\ref{eq:influ_est}) exactly corresponds to Approach 1 in Theorem~\ref{sec:theory}. \qed

All quantities in Table~\ref{tab:QOI} can be computed using this result. Consider two examples:

\textbf{Example 1 (Average Marginal Indirect Effect)}: $\hat{\delta}(1) = \hat{\alpha}(1,1)-\hat{\alpha}(1,0)$ can be estimated using the following linear regression
\begin{equation*}
\hat{\delta}(1) = \hat{\delta}_1^{\texttt{OLS}} \quad \hat{\delta}^{\texttt{OLS}}_0, \hat{\delta}^{\texttt{OLS}}_1 = \argmin_{d_0, d_1} \sum_{i=1}^N \sum_{t' \in \{0, 1\}} \left(\psi_i(1,t') - d_0 - d_1 \I{t' = 1}  \right)^2.
\end{equation*}

The proof follows by application of Corollary~\ref{coro:est_avg} with $\bm{z}$ that has a `1` in the position of $\psi_i(1,1)$ and a `-1` in the position $\psi_i(1,0)$. Note that Approach 3, excluding the columns of $\bm{Y}$ where $\bm{z}$ equals zero is equivalent to a linear regression where the linear predictor is parameterized as $\beta_0 \I{t=t,t'=0} + \beta_1\I{t=t,t'=1}$. Then, $\bm{z}^T\hat{\bm{\beta}} = \hat{\beta}_1 - \hat{\beta}_0$. In this example, the multiplication by $\bm{z}$ can be achieved instead by re-parameterizing as in the above equation and noting that $\hat{\delta}^{\texttt{OLS}}_1 = \bm{z}^T \hat{\bm{\beta}}$.

\textbf{Example 2 (Marginal Average Marginal Indirect Effect):} The paper defines $\bar{\delta}$ as $1/|\mathcal{T}|\sum_{t\in \mathcal{T}} \delta(t)$ (Table~\ref{tab:QOI}) and reports this in Section~\ref{sec:empirical}. This can be estimated by the specific choice of $\bm{z}$ noted before. However, Section~\ref{sec:estimation} states it can be estimated by the following linear regression:
\begin{equation}\label{eq:app_example_2}
\hat{\bar{\delta}} = \hat{\bar{\delta}}_1^{\texttt{OLS}} \quad \hat{\bar{\delta}}^{\texttt{OLS}}_0,  \hat{\bar{\delta}}^{\texttt{OLS}}_1 = \argmin_{\bar{d}_0, \bar{d}_1} \sum_{i=1}^N \sum_{t \in \{0, 1\}} \sum_{t' \in \{0, 1\}} \left(\psi_i(1,t') - \bar{d}_0 - \bar{d}_1 \I{t' = 1}  \right)^2.
\end{equation}

This again follows from Corollary~\ref{coro:est_avg} and the assumption of a balanced design. To establish this, we use Lemma~\ref{lemma:balanced}, below:

\begin{lemma}\label{lemma:balanced}
    Assume a dataset with $N$ units, $i \in \{1, \cdots N\}$ and two factors $g \in \{1,\cdots, G\}$ and $f \in \{1, \cdots, F\}$. For each unit $i$, we observe outcomes $\bm{y}_i$ that contains one observation for each $(g,f)$ combination. Consider two linear regressions:
    \begin{itemize}
        \item Approach A: Predict $\bm{y}_i$ using a design matrix $\bXi = \left(\bm{I}_G \otimes \bm{1}_F\right)$ that contains one-hot encoding for each level of $g$ corresponding to the elements of $\bm{y}_i$. Stacking the designs and data into $\bm{X}$ and $\bm{y}$, define $\hat{\bm{\beta}}_A = \left(\bm{X}^T\bm{X}\right)^{-1}\bm{X}^T\bm{y}$.
        \item Approach B: Predict $\bm{y}_i$ using a fully saturated model where a design matrix $\bm{W}_i = \bm{I}_{FG}$ where each row corresponds to a one-hot encoding for each combination of $(g,f)$. Define $\hat{\bm{\gamma}} = \left(\bm{W}^T\bm{W}\right)^{-1} \bm{W}^T\bm{y}$ where $\bm{W}$ and $\bm{y}$ are stacked across units $i$. Define $\hat{\bm{\beta}}_B$ as the averaged values of $\hat{\bm{\gamma}}$ within the corresponding strata, i.e.  $[\hat{\bm{\beta}}_B]_g = \frac{1}{F}\sum_{f=1}^F \hat{\bm{\gamma}}_{f,g}$. Define $\bm{M}$ as a $G \times FG$ matrix that does this, i.e. $\hat{\bm{\beta}}_B = \bm{M} \hat{\bm{\gamma}}$ where $\bm{M} = \left(\bm{I}_G \otimes \frac{1}{F} \bm{1}_F^T\right)$.
    \end{itemize}

    Then, $\hat{\bm{\beta}}_A= \hat{\bm{\beta}}_B$. If HC0 cluster-robust standard errors on each person $i$ are used to obtain $\hat{\bm{\beta}}_A$ and $\hat{\bm{\gamma}}$, then $\widehat{\mathrm{Cov}}(\hat{\bm{\beta}}_A) = \widehat{\mathrm{Cov}}(\hat{\bm{\beta}}_B)$.
\end{lemma}

\emph{Proof:} The equivalence of the point estimates $\hat{\bm{\beta}}_A = \hat{\bm{\beta}}_B$ follows from elementary book-keeping and the balanced design. The equivalence of the estimated covariance matrix starts by noting
$$
\widehat{\mathrm{Cov}}(\hat{\bm{\beta}}_A) = \left(\frac{NF}{N}~\bm{I}_{G}\right)^{-1} \left(\sum_{i=1}^{N} \bXi^T \left(\bm{y}_i - \bXi\hat{\bm{\beta}}_A\right)\left(\bm{y}_i - \bXi\hat{\bm{\beta}}_A\right)^T\bXi\right) \left(\frac{NF}{N}~\bm{I}_G\right)^{-1}.
$$

Further, for $\hat{\bm{\gamma}}$,
$$
\widehat{\mathrm{Cov}}(\hat{\bm{\gamma}})=\left(\frac{N}{N}~\bm{I}_{FG}\right)^{-1}\left(\sum_{i=1}^N \left(\bm{y}_i - \hat{\bm{\gamma}}\right)\left(\bm{y}_i - \hat{\bm{\gamma}}\right)^T\right)\left(\frac{N}{N}\bm{I}_{FG}\right)^{-1}
$$

For $\hat{\bm{\beta}}_B$, $\widehat{\mathrm{Cov}}(\hat{\bm{\beta}}_B) = \bm{M}\widehat{\mathrm{Cov}}(\hat{\bm{\gamma}})\bm{M}^T$ and note that $\bm{M} = \frac{1}{F} \bXi^T$, thus
\begin{align*}
\widehat{\mathrm{Cov}}(\hat{\bm{\beta}}_B) &= \frac{1}{F^2} \left(\sum_{i=1}^N \bXi^T \left(\bm{y}_i - \hat{\bm{\gamma}}\right)\left(\bm{y}_i - \hat{\bm{\gamma}}\right)^T \bXi\right).
\end{align*}

To show that $\widehat{\mathrm{Cov}}(\hat{\bm{\beta}}_A)=\widehat{\mathrm{Cov}}(\hat{\bm{\beta}}_B)$, note that $\bXi^T\bXi \hat{\bm{\beta}}_A = F \hat{\bm{\beta}}_A$ and $\bXi^T\hat{\bm{\gamma}} = F \hat{\bm{\beta}}_B$. As $\hat{\bm{\beta}}_A=\hat{\bm{\beta}}_B$, these coincide as do the estimates of the covariance matrices. \qed

To establish Example 2, by Lemma~\ref{lemma:balanced}, the regression estimates (and standard errors) from Equation~\ref{eq:app_example_2} are equivalent to estimating the model with all interactions and then averaging the relevant coefficients (i.e., Approach 3 in Theorem~\ref{thm:estimation}). Thus, by Collorary~\ref{coro:est_avg}, Equation~\ref{eq:app_example_2} is equivalent to direct estimation using Approach 1. Simple argumentation also shows that these results apply when $t$ has more than two levels.

\subsection{Heterogeneous Effects}\label{app:deriv_heteff}

Section~\ref{sec:heteff} notes that the best linear approximation to the conditional expectation function can be found using linear regression. To establish this, we define the conditional expectation effect for any $\bm{z}$ as    $F(\bm{z};\bx) = \E[\bm{z}^T\bm{\psi}_i | \bXi = \bx]$. Assume there is a low-dimensional set of covariates $\bm{w}_i$---that we assume to be a subset or other deterministic transformation of $\bm{x}_i$---for which we wish to obtain the best linear approximation to $F(\bm{z}; \bx)$. Following results from \cite{semenova2021debiased}, this can be estimated by a linear regression to predict $\bm{z}^T\hat{\bm{\psi}}_i$ using $\bm{w}_i$.\footnote{It can also be shown, somewhat tediously, that an identical estimator can be obtained using results from \cite{kennedy2024semiparametric} to find the influence function of the corresponding population linear regression coefficient of the conditional effect function as the outcome and $\bm{w}_i$ as the predictors.} In simple settings, e.g., with a binary $W_i$, this collapses to the difference in average $\bm{z}^T\hat{\bm{\psi}}_i$ when $W_i = 1$ from when $W_i = 0$.

The estimator that we use, $\hat{\bm{\beta}}_{\bm{z}}$, is defined as shown below, noting again that this is a consistent and asymptotically normal estimator for $\bm{\beta}_{\bm{z}}$, where $\hat{\bm{\Psi}} \in \mathbb{R}^{N \times q}$ is a matrix of stacked $\hat{\bm{\psi}}_i^T$:\footnote{Unlike \cite{semenova2021debiased}, we fix $\bm{w}$ and thus this will not recover the \emph{true} conditional expectation function asymptotically, rather, we focus on obtaining the best linear approximation for easy interpretability. This also requires milder assumptions than in their paper.}

\begin{equation}
\begin{aligned}
\hat{\bm{\beta}}_{\bm{z}} &= \left(\bm{W}^T\bm{W}\right)^{-1}\bm{W}^T\hat{\bm{\Psi}} \bm{z}; \\
\widehat{\mathrm{Var}}\left(\hat{\bm{\beta}}_{\bm{z}}\right) &= \left(\frac{1}{N}~ \bm{W}^T\bm{W}\right)^{-1} \left(\frac{1}{N}\sum_{i=1}^N (\bm{z}^T\hat{\psi}_i - \hat{\bm{\beta}}_{\bm{z}}^T\bm{w}_i)^2 \cdot \bm{w}_i\bm{w}_i^T\right)\left(\frac{1}{N}~\bm{W}^T\bm{W}\right)^{-1} \\
\sqrt{N}\left(\hat{\bm{\beta}}_{\bm{z}} - \bm{\beta}_{\bm{z}}\right) &\to^d N\left(\bm{0},\widehat{\mathrm{Var}}(\hat{\bm{\beta}}_{\bm{z}})\right)
\end{aligned}
\end{equation}

Theorem~\ref{thm:estimation} applies here immediately and thus provides three equivalent ways of estimating heterogeneous effects. The conditional effect function for all mediation quantities of interest can be obtained for specific choice of $\bm{z}$ as above. To consider one specific example,

\textbf{Example 3 (Heterogeneous Indirect Effects):} Equation~\ref{eq:app_het} states that heterogeneous effects for a binary moderator $W_i \in \{0,1\}$ can be obtained using the following regression

\begin{equation}\label{eq:app_het}
\begin{split}
&(\hat{\delta}_0,\hat{\delta}_W, \hat{\delta}_{1}, \hat{\delta}_{1:W})\\
&= \argmin_{d_0, d_W, d_1, d_{1:W}} \sum_{i=1}^N \sum_{t' \in \{0,1\}} \left(\begin{split}&\hat{\psi}_i(1,t') - d_0 - d_W \I{W_i = 1} - d_1 \I{t' = 1} \\ &-d_{1W} \I{t'=1,W_i=1}\end{split}\right)^2
\end{split}
\end{equation}

in which the interaction term $\hat{\delta}_{1:W}$ represents the difference in average marginal indirect effects when $W_i = 1$ versus $W_i = 0$. The proof is as follows: Theorem~\ref{thm:estimation} states three equivalent ways of estimating $\delta(1;1)-\delta(1;0)$. Approach 1 in this instance provides a single regression estimator that is equivalent to a single regression with linear predictor as follows, noting that since the sum in Equation~\ref{eq:app_het} runs only over those $\hat{\psi}_i(t,t')$ where $t = 1$, we can exclude all other elements of $\bm{z}$ from the regression:
\begin{equation*}
\begin{split}
&\I{t=1,t'=0} \gamma_1 + \I{t=1,t'=1} \gamma_2 + \\
&\I{t=1,t'=0,W_i=1} \gamma_3 + \I{t=1,t'=1,W_i=1} \gamma_4.\end{split}\end{equation*}
This can be re-parameterized in terms of differences, i.e. where $d_0 = \gamma_1$, $d_W = \gamma_3 - \gamma_1$, $d_1 = \gamma_2 - \gamma_1$, $d_{1:W} = (\gamma_4 - \gamma_3) - (\gamma_2 - \gamma_1)$. Thus, Theorem~\ref{thm:estimation} guarantees the equivalence between Equation~\ref{eq:app_het} and the direct estimator. Tedious, if straightforward, algebra shows that a similar result applies to quantities such as heterogeneous averaged indirect effects, i.e. the heterogeneous effects extension of Example 2 above.

\subsection{Repeated Observations}

As Section~\ref{sec:extensions} notes, one often has individual $i$ perform multiple comparison tasks. Using the linear regression formulation for estimation in Theorems 1 and 2, no adjustment is needed; clustered standard errors for individual $i$ will account for dependence across repeated observations and different $(t,t')$-combinations.

\section{Sensitivity Analysis}\label{app:sensitivity}

Our method depends on two critical, untestable, assumptions: Principal ignorability (Assumption~\ref{eq:A_PI}) and manipulation exclusion restriction (Assumption~\ref{eq:A_ExcluRes}). Given those assumptions plus others that can be guaranteed to hold by design---i.e., random assignment, Equation~\ref{eq:app_med_form} (repeating Equation~\ref{eq:mediation_fmla} from the main text) identifies our key quantity $\alpha(t,t') = \E[Y_i(t,M_i(t',\bSi),\bSi)]$ as follows
\begin{equation}\label{eq:app_med_form}
\alpha(t,t') =  \sum_{x,\bs} \left[\sum_{m} \mu_Y(t,m,\bx, \bs) e_m(t', \bs,\bx)\right]\mathrm{Pr}(X_i = x, S_i= s),    
\end{equation}

where $\mu_Y(t,m,\bx, \bs) = \E[Y_i | T_i = t, M_i = m, \bXi = \bx, \bSi = \bs, A_i = 1]$ and $e_m(t',\bs,\bx) = \mathrm{Pr}(M_i = m | T_i = t', \bSi = \bs, \bXi = \bx, A_i = 0)$.

This section derives a sensitivity test to examine the robustness of our results to violations of the key untestable assumptions. It also includes some additional results from applying this to our empirical analysis.

\subsection{Sensitivity Analysis in Mediation}

We follow the logic of existing sensitivity analyses in semi-parametric settings (e.g., \citealt{robins2000sensitivity,tchetgen2012semiparametric}) by supposing the existence of a function $\gamma(t,t',m,\bx,\bs)$ that captures the magnitude of the bias that is driven by (i) differences in average potential outcomes between principal strata and (ii) differences in average potential outcomes due to direct manipulation of the mediator. Formally,
\begin{equation}
    \begin{split}
    \gamma(t,t', m, \bx, \bs) = ~&\E[Y_i(t,m) | \I{M_i(t) = m}, X_i = x, A_i = 1]  - \\ 
    &\E[Y_i(t,m) | \I{M_i(t') = m}, X_i = x] 
    \end{split}
\end{equation}

Simple algebraic re-arrangement shows that the following decomposition holds, \emph{without} assuming principal ignorability or a manipulation exclusion restriction. The only non-verifiable assumption is that $e_m(t',\bx,\bs) = \mathrm{Pr}(M_i = m|T_i = t, \bXi = \bx, \bSi = \bs)$; this holds if $M_i(t,\bs,a) = M_i(t,\bs,a')$ and subjects are randomized into experiments.

$$\E[Y_i(t,t')] = \alpha(t,t') - \sum_{\bx,\bs} \left[\sum_{m} \gamma(t,t',m,\bx, \bs)e_m(t',\bx,\bs)\right]\mathrm{Pr}(\bXi = \bx, \bSi = \bs)$$

If we assumed some specific functional form on $\gamma(\cdot)$, we could assess sensitivity. It is, however, difficult to think through a plausible functional form. To address this, we leverage the (optional) $Y(T)$-experiment where $T_i$ is manipulated but $M_i$ takes on its natural value to provide some information on $\gamma(\cdot)$. In this experiment given the randomization of $T_i$, we know that $\E[Y_i(t,t)]$ is identified by
\begin{equation}
\begin{split}
\phi(t) &= \sum_{\bx,\bs} \mu_Y^0(t, \bx, \bs)\mathrm{Pr}(X_i = x, \bSi = s) \\
\mu_Y^0(t,\bx,\bs) &= \E[Y_i | T_i = t, \bXi = \bx,\bSi = \bs, A_i = *]   
\end{split}
\end{equation}

If all of our assumptions hold, then $\phi(t) = \alpha(t,t)$ as well as the conditional expectation holding for all $(\bXi,\bSi)$ strata. This suggests a way to find an approximation for $\gamma(\cdot)$: Find the function that minimizes the expected squared difference between the conditional expectation from the $Y(T)$-experiment and the one that results from the mediation formula, summing over all values of $t$. 

\begin{equation}
    \begin{split}
    &\tilde{\gamma}(t,t,m,\bx,\bs) = \argmin_{\gamma} \sum_{t,\bx,\bs} \left[\begin{split}&\E[Y_i|T_i = t, \bXi = \bx, \bSi = \bs, A_i = *]\\
    &- \sum_m \mu_Y(t,m,\bx, \bs)e_m(t,\bx,\bs) \\
    &- \sum_m \gamma(t,t,m,\bx,\bs)\end{split}\right]^2  \mathrm{Pr}\left(\begin{split}&T_i = t, \\ &\bXi = \bx, \\ &\bSi = \bs\end{split}\right) 
    \end{split}
\end{equation}

With the data provided by the $Y(T)$-experiment, however, we can only estimate $\gamma(\cdot)$ when $t=t'$, as the $Y(T)$-experiment gives no information about nested counterfactuals. To address this and make predictions for $\gamma(\cdot)$ when $t \neq t'$, we assume that it is symmetric in its arguments, i.e. $\gamma(t,t',m,\bx,\bs)=\gamma(t',t,m,\bx,\bs)$. Further, instead of the challenging task of estimating the best $\gamma(\cdot)$, we focus merely on a linear approximation in $M$, $X$ and $S$ where, additionally, we include interactions with $t$ and $t'$, i.e.

\begin{equation}
    \begin{split}
    \gamma(\cdot) &= \sum_{u \in \mathcal{T}} \left(\I{u = t} + \I{u = t'}\right) \alpha_u + \bx^T\bm{\beta}_X + \bs^T\bm{\beta}_S + \sum_{k} \I{k = m} \beta_m + \\
    & \sum_{u \in \mathcal{T}} \left(\I{u = t} + \I{u = t'}\right) \left[\bx^T \bm{\phi}^{(u)}_{X} + \bs^T \bm{\phi}^{(u)}_{S} + \sum_{k} \I{k=m} \phi^{(u)}_{m}\right] = \bm{w}^T \bm{\theta}
    \end{split}
\end{equation}

Thus, we can find an approximation to $\gamma(\cdot)$ by solving the implied linear system above.

To do so, we estimate the high-dimensional functions $\mu_Y(\cdot), e_m(\cdot), \E[Y_i | T_i = t, \bXi = \bx, \bSi = \bs, A_i = *]$ using cross-fitting. That is, we split the data into $K$ folds, estimate the model on all but one fold, and plug-in the out of sample predictions on the observed $(T_i, \bXi, \bSi)$ pairs. We then estimate $\hat{\bm{\theta}}^*$ to define the bias function $\hat{\gamma}(\cdot)$. Future work might look into estimating this in a doubly robust fashion.

\subsection{Application of Sensitivity Function}

Section~\ref{sec:partial_test} applies this estimated sensitivity function to the experiment and shows how this changes the estimated effects. Figure~\ref{fig:sens_mm} complements this analysis by showing how applying the sensitivity analysis changes the estimates of $\E[Y_i(t,M_i(t,\bSi),\bSi)] = \alpha(t,t)$. As expected, even though we use a simple linear approximation to estimate $\gamma(\cdot)$, this adjustment resolves many of the observed discrepancies between the estimates coming from the mediation formula and the $Y(T)$-experiment itself---especially when race is used as the primary treatment $T$. The linear approximation is perhaps too crude for political experience insofar as it \emph{over-corrects} beyond the means observed in the $Y(T)$-study.

\begin{figure}
    \caption{Sensitivity of Marginal Means}\label{fig:sens_mm}
    \includegraphics[width=\textwidth]{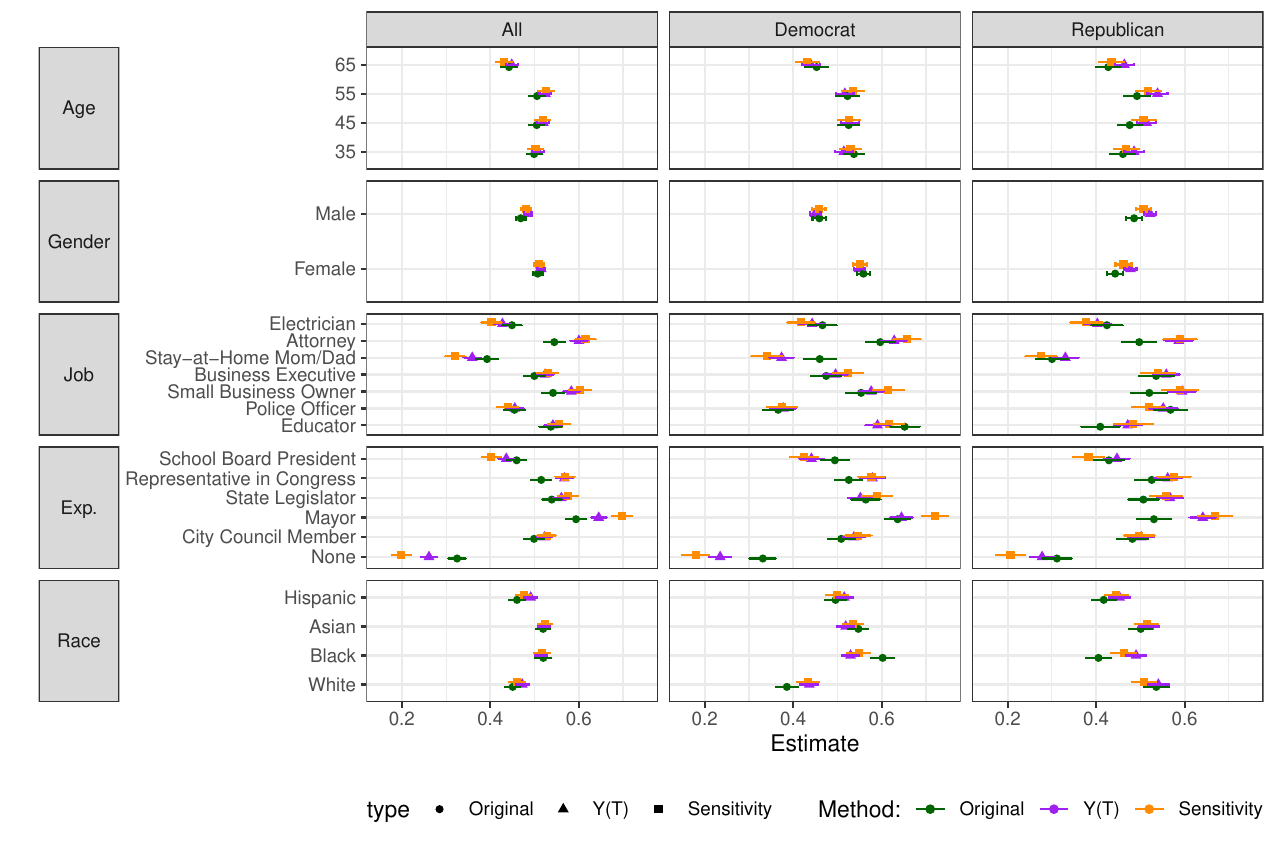}
\end{figure}

\section{Additional Empirical Results}\label{app:other_empirical}

This section contains other empirical results including the estimated eliminated effects from our analyses, full results for all respondents and each partisan sub-group, and results for other estimands.

\subsection{Eliminated Effects}\label{app:emp_eliminated}

Figures~\ref{app:repl_dem} and~\ref{app:repl_rep} show the estimated effects from the $Y(T)$ and $Y(T,M)$-experiments, i.e., the AMCEs, as well as the difference between them, i.e., the eliminated effect \citep{acharya_2018_polan}. These are all identified without requiring principal ignorability, although the eliminated effect does require the manipulation exclusion restriction. To benchmark our results, we also present the corresponding estimates from the data used in \cite{kirkland_candidate_2018}'s two studies---one on MTurk and one on YouGov. Our study is denoted as ``Prolific.''

\begin{figure}
\caption{Democratic Respondents: $Y(T)$, $Y(T,M)$ and Eliminated Effects}\label{app:repl_dem}
\includegraphics[width=\textwidth]{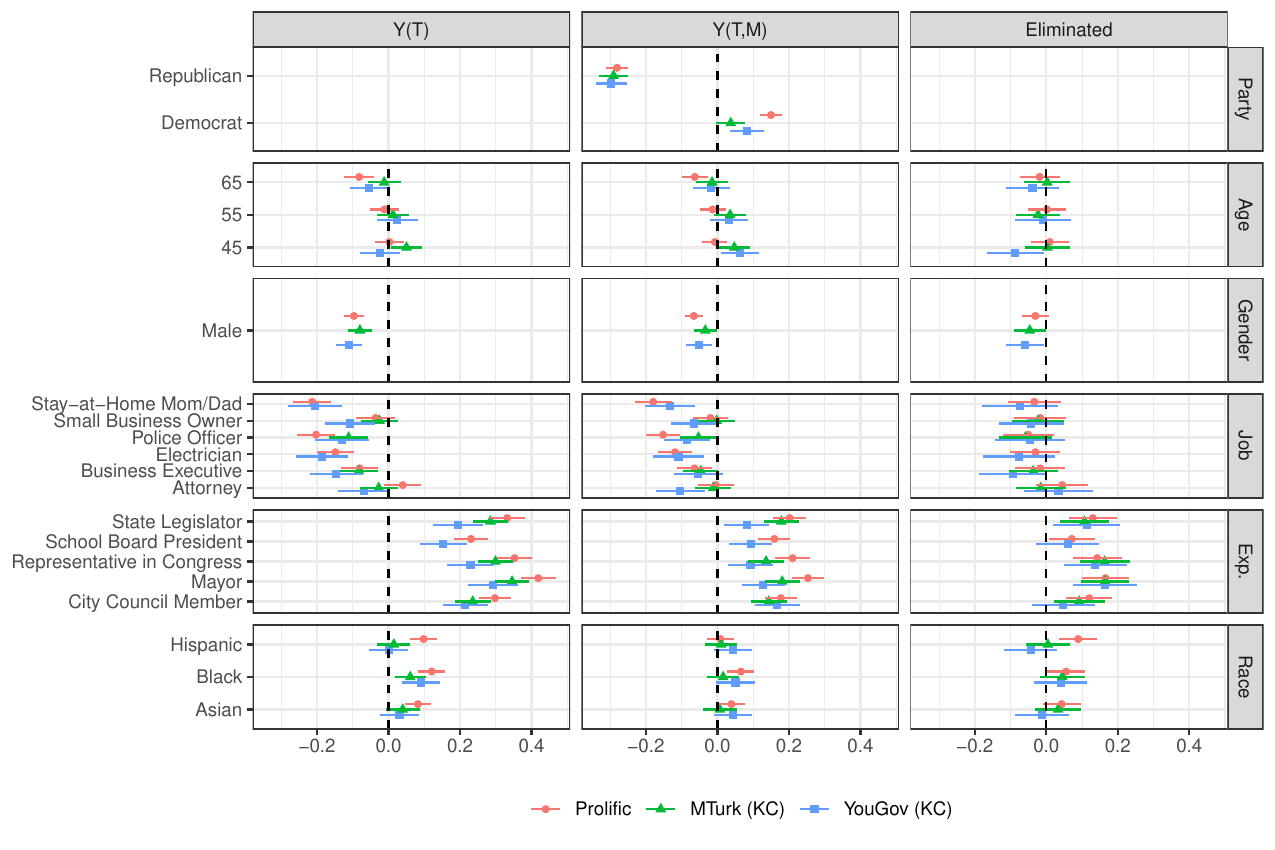}
\end{figure}

\begin{figure}
\caption{Republican Respondents: $Y(T)$, $Y(T,M)$ and Eliminated Effects}\label{app:repl_rep}
\includegraphics[width=\textwidth]{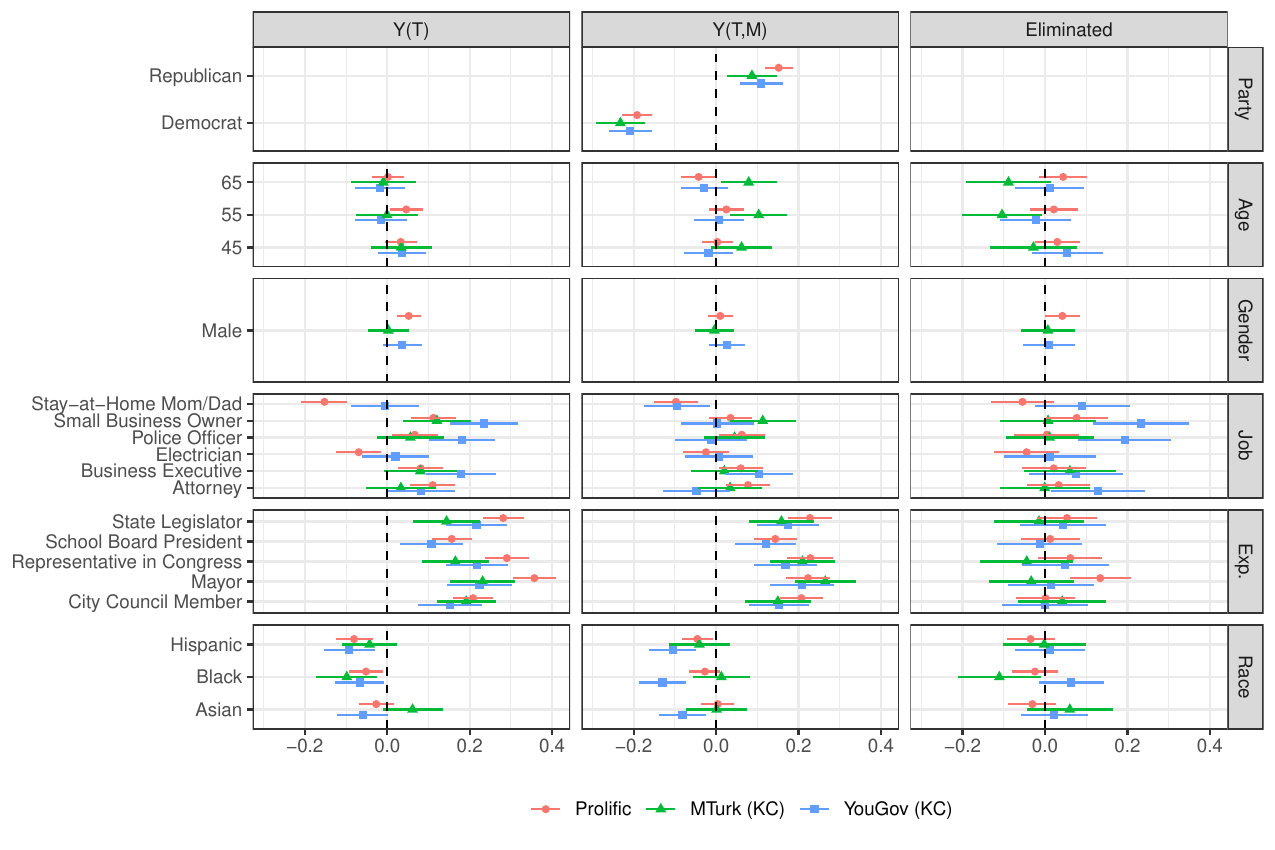}
\end{figure}

Across the studies, the results are generally similar. The correlation of the point estimates in both the $Y(T,M)$-studies and $Y(T)$-studies (across all effects and partisan sub-groups) across Prolific estimates and the YouGov and MTurk studies is above 0.90. More differences exist in the eliminated effects, although the estimates from our study are as correlated with MTurk (0.59) and YouGov (0.49) as those two studies are with each other (0.45).

Focusing on the eliminated effect, we note that our study (``Prolific'') finds some statistically significant evidence for their existence. Focusing on Black vs. white candidates, the eliminated effect is estimated to be 0.055 ($p$-value of 0.04) for Democratic respondents and -0.024 ($p$-value of 0.39) for Republican respondents. By contrast, the mediation analysis provides clear evidence---even once the sensitivity analysis is included---of indirect effects. If desired, our framework could allow one to estimate the four-fold decomposition by \cite{vanderweele_2014_epidem} and interpret this reference interaction directly.

\subsection{Full Sample and Partisan Sub-Group Results}\label{app:full_results}

For completeness, we replicate the main results (effects from the $M(T)$-experiment and showing the MAMIE, MAMDE, and AMCE) by all respondents and three partisan categories (Democrat, Republican, and Independent). Independents are more noisily estimated, as there are fewer of them in the sample, but are generally located between the Democratic and Republican respondents. The results from the $M(T)$-experiment are shown first. There is limited evidence of heterogeneous effect by the party of the respondent.

\begin{figure}[!htbp]
    \caption{Estimated Effects from $M(T)$-Experiment}
    \includegraphics[width=\textwidth]{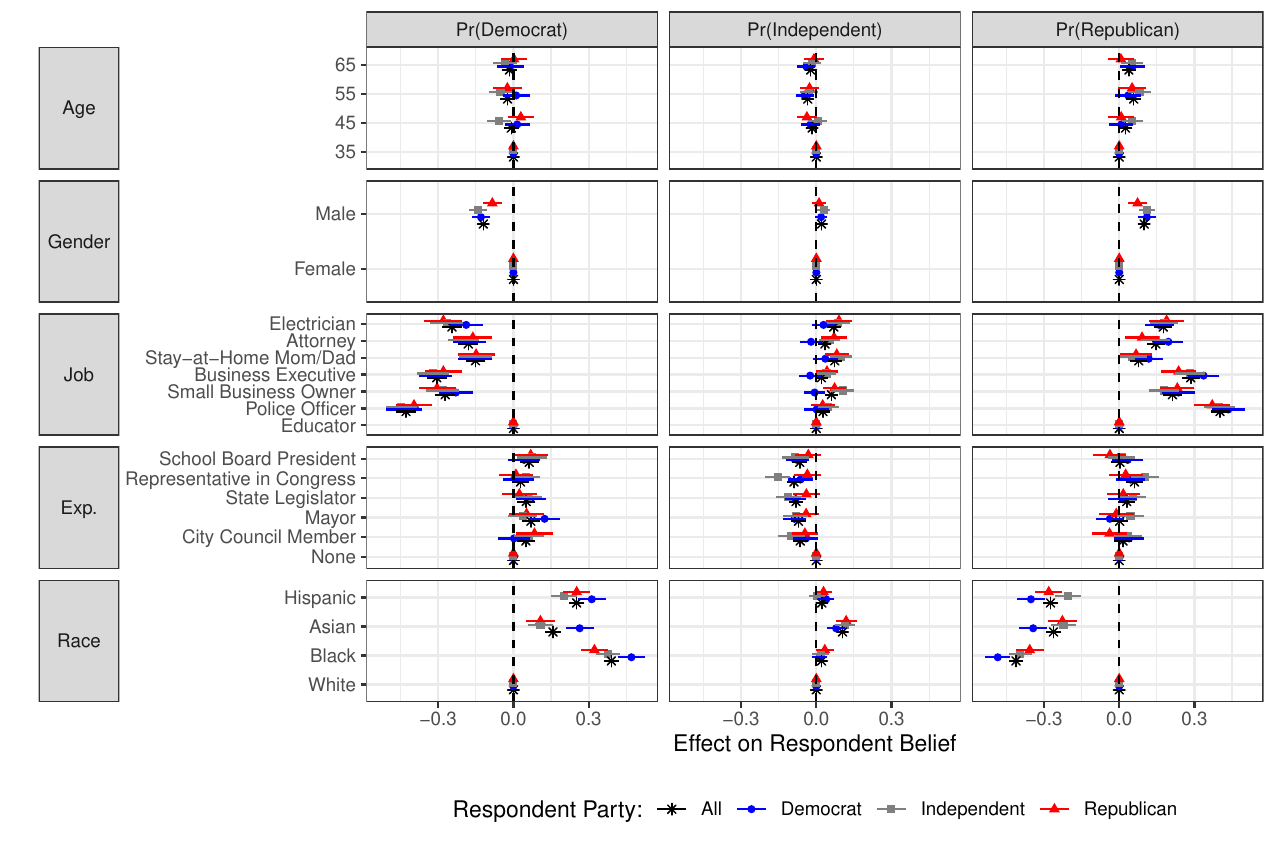}
\end{figure}

Next, Figure~\ref{app:main_effects_robust} replicates the main figure estimating the mediation effects. The interpretation of the third pre-registered hypothesis (indirect effect of political experience for all respondents) is as follows: As in the main text, the bulk of the effect is through a direct effect. The indirect effects are statistically significant for all levels except Mayor ($p$-values from 0.0005 to 0.005) but quite small in magnitude.

\begin{figure}
    \caption{Estimated Mediation Effects}\label{app:main_effects_robust}
    \includegraphics[width=\textwidth]{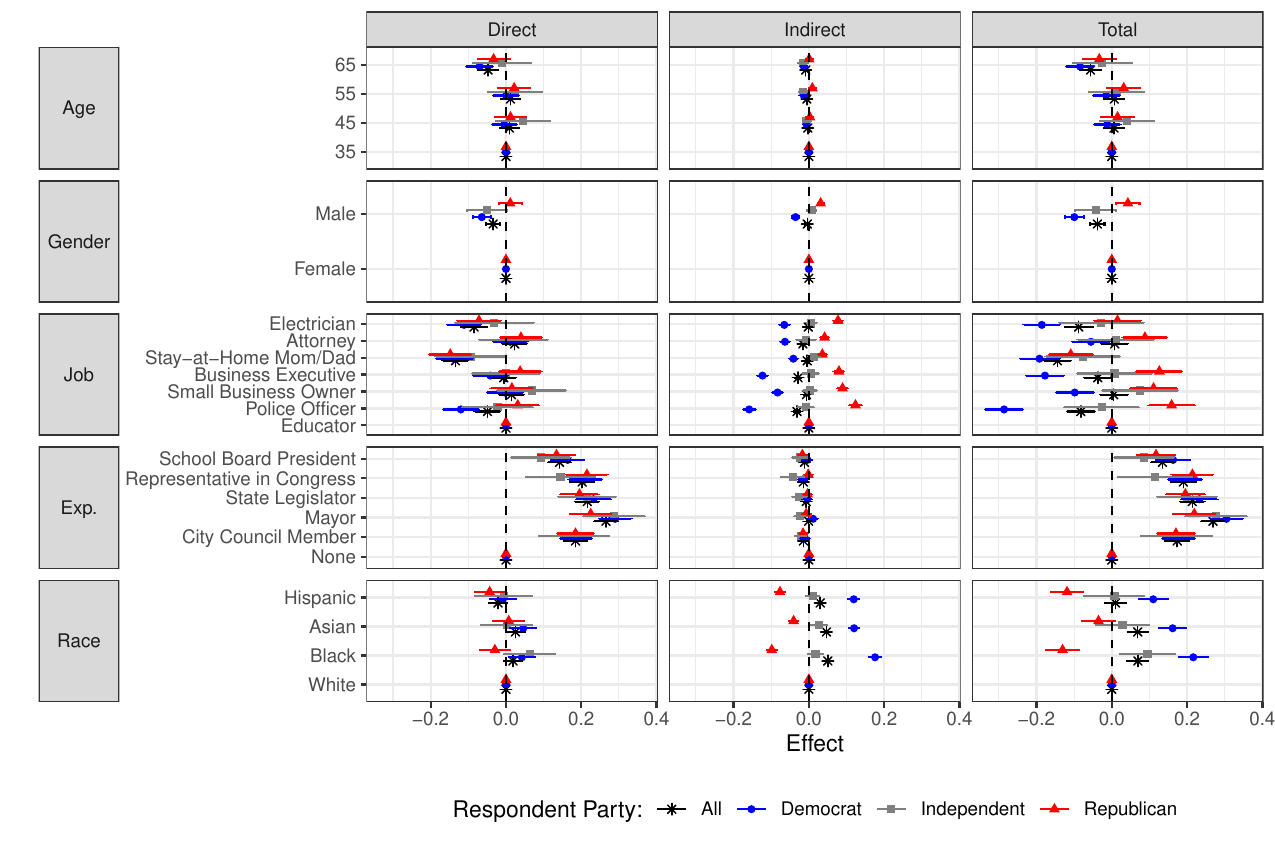}
\end{figure}

\subsection{Average Marginal Direct and Indirect Effects}\label{app:extra_a_effects}

As noted in the main text, for a treatment with levels $\mathcal{T}$, Table~\ref{tab:QOI} notes that there are average marginal direct (AMDE) and indirect effects (AMIE) for each level $t \in \mathcal{T}$. This is a large number of quantities to examine, so we prefer to report the average of these, i.e., the marginalized average marginal direct and indirect effects (i.e., MAMDE and MAMIE). Figure~\ref{fig:app_extra_A} reports each of these for Democratic and Republican respondents, in light blue or light red, where the corresponding MAMIE and MAMDE from the main text are overlaid in dark blue and dark red, respectively. There is some variation (e.g., for indirect effects for Job), but the broad story is identical to the main text. 

\begin{figure}[!htbp]
    \caption{AMIE and AMDE by Partisanship of Respondent}\label{fig:app_extra_A}
    \includegraphics[width=\textwidth]{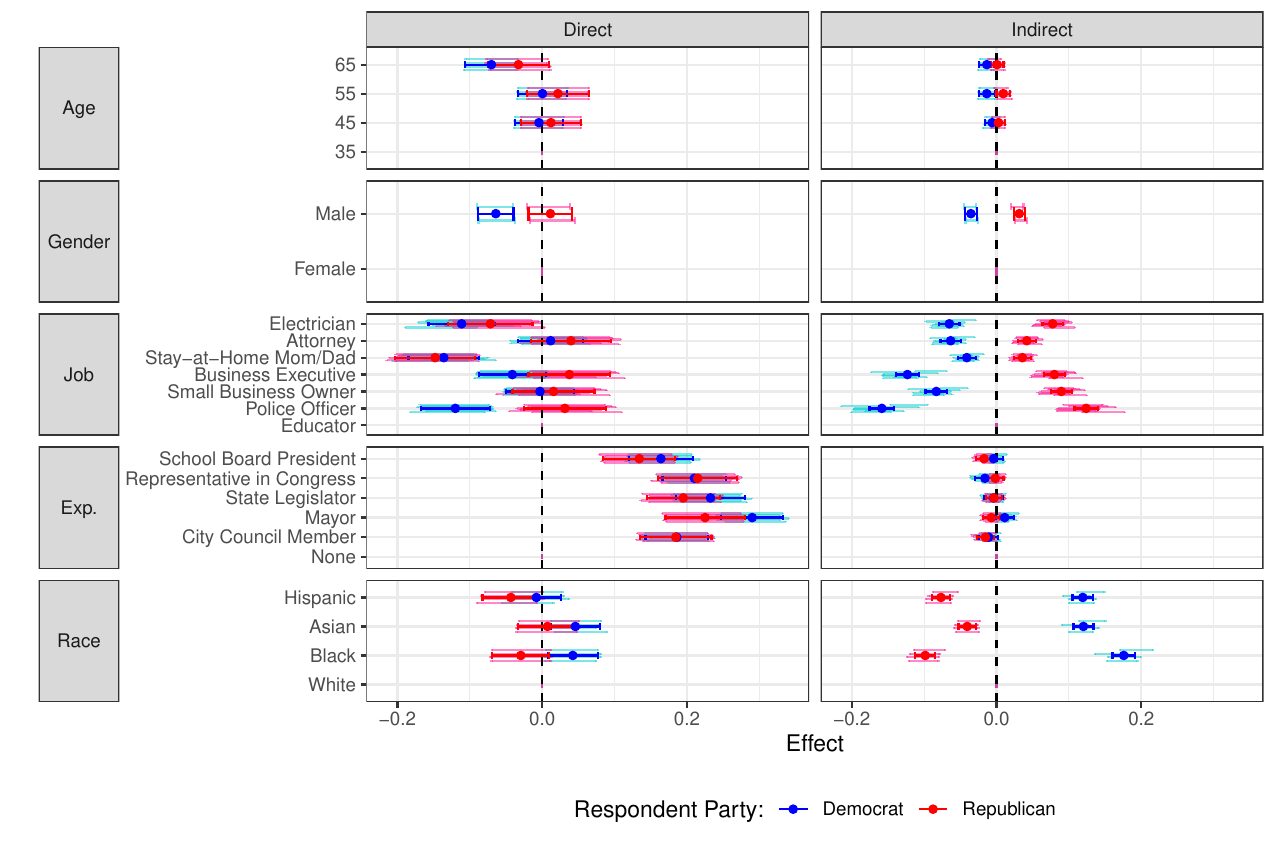}
\end{figure}

\section{Additional Information on the Survey}\label{app:survey_details}

We include here additional information on the survey itself and recruitment of respondents. Appendix~\ref{app:survey_instrument} contains a link to the full survey instrument. 

\subsection{Recruitment of Respondents}\label{app:survey_demo}

Our pre-registration plan on OSF (available \chref{\urlprereg}{here}) said we would recruit 4,500 respondents. We fielded the survey from August 29, 2025-October 19, 2025 and recruited 4,549 respondents, with at least 1,500 in each experiment. All respondents were above 18 and located in the United States. Per our IRB and Prolific's guidelines, we paid each respondent \$1.50 based on an estimated completion time of 7.5 minutes. In fact, the median respondent took 3.7 minutes with a median time in the $Y(T,M)$ and $Y(T)$-experiments of slightly over four minutes and the $M(T)$-experiment taking around 2.7 minutes.

To obtain a reasonable number of Republican respondents, we used Prolific's quotas to recruit a more balanced sample of partisans (29.5\% Democrat; 27.9\% Republican; 42.6\% Independent), with quotas also for age and sex. However, the self-reported partisanship of the respondents differ from the information in Prolific's quota when measured using the standard three-point partisanship scale (e.g., 47.2\% of respondents self-identify as Democrats). The following table compares the quota to self-reported partisanship used in the main text, where independent ``leaners'' are grouped with their corresponding party.

\begin{table}[!htbp]
    \caption{Partisanship Distribution: Self-Reported vs. Prolific}
    \centering
    \begin{tabular}{llllll}
    \hline\hline
    && \multicolumn{3}{c}{Self-Reported Party ID}  & \\
    \hline
    & & Democrat & Independent & Republican & Margin  \\
    \input{figures/tab_quota.tex} \\ \hline\hline
    \end{tabular}
\end{table}

The pre-registration plan lists the pre-treatment covariates we measured including age, gender, education, income, ethnicity, Hispanic identity, interest in politics, ideology, partisanship (measured on both a seven-point and three-point scale). These were measured to mirror either the coding in \cite{kirkland_candidate_2018} or standard question coding for political surveys in the United States (see Appendix~\ref{app:survey_instrument}) for the full instrument. In both the respondent pre-treatment covariates and the experimental responses, there was very little missing data. 

\subsection{Allocation of Respondents to Experimental Arms}

Each respondent was randomly assigned with 1/3 probability into each of the three experiments, i.e., $Y(T)$, $Y(T,M)$, and $M(T)$. In each experiment, each respondent performed five tasks described in the main text. In the $Y(T)$ and $Y(T,M)$-experiments, these were forced choice conjoints (i.e., comparing two candidates). In the $M(T)$-experiment, they provided their best guess as to the party affiliation of a single candidate. As all of our pre-treatment demographic variables are categorical and there are three experimental arms, we conduct a chi-squared test for each measured pre-treatment demographic (and a recoded version if used in the regression analysis) to see if there is any evidence of imbalance. Across all of the demographic factors, the smallest $p$-value is around 0.23.

\subsection{Machine Learning Models}\label{app:survey_ML}

Our method requires estimating high-dimensional conditional expectation functions for the outcome, i.e. $\mathbb{E}[Y_i | T_i = t, \bSi = \bm{s}, \bXi = \bm{x}, M_i = m, A_i = 1]$, and the mediator, i.e., $\mathrm{Pr}(M_i = m | T_i = t, \bSi = \bm{s}, \bXi = \bm{x}, A_i = 0)$. To do this, we use the following machine learning models. Specifically, we rely on a random forest estimated with \texttt{ranger}, using a probability forest with 500 trees (the default setting). Our predictive models use the full array of our pre-treatment covariates (partisanship [7 point], ideology, gender, household income, ethnicity, education, interest in politics, age), with some recoding to collapse rare categories; Appendix~\ref{app:survey_instrument} provides the question wording.

We rely on a random forest for two reasons; first, it can handle the small amount of missing data observed. Second, it is a common option used in debiased/double machine learning because of its flexibility in learning interactions, strong out-of-sample predictive performance, and convergence rates. \cite{montgomery2018tree} provide an accessible introduction to this method for political science data. This is one option that we suggested in the pre-analysis plan. As all of our features are categorical, we also prefer \texttt{ranger} as it has the ``partition'' argument that splits factors based on the optimal partition of levels (e.g., grouping together Democratic and Independent respondents separate from Republican ones). This is important as the ``standard'' implementation is likely to work poorly with categorical inputs. Although they are known by design, to deal with possible random imbalances, we estimate the propensity scores using a simpler model---an additive multinomial or logistic regression. For $p^T_t$, we predict this using $\bSi$; for $p^M_m$, we predict this using $T_i, \bSi$; for $p_a$, we predict this using $\bXi$. For $p_a$, we use a multinomial logistic regression with random effects as some categories in $\bXi$ are rather rare in our sampled data.

We rely on five-fold cross-fitting to estimate the nuisance functions and obtain $\hat{\psi}_i(t,t')$ using only held-out data. For each fold $k$, we estimate the conditional expectation functions once and then compute all $\hat{\psi}_i(t,t)$ using those same estimated models. In our analysis presented here, we use the empirical distribution of $\bSi$ and thus rely on the influence function discussed in the main text. Appendix~\ref{app:mediation_fmla} discusses how one might weight to a different distribution following \cite{de2022improving}.

\subsection{Attention Check}

We implemented a simple attention check in the survey that asked respondents to ``please select the second option from those below'' from ``strongly disagree'', ``somewhat disagree'', ``neither agree nor disagree'', ``somewhat agree'', or ``strongly agree''. Only 1.3\% of respondents failed this attention check. As we pre-registered that we would run the analysis excluding those respondents, our replication code does so. Focusing on Figure~\ref{fig:main_results}, the estimates are virtually identical, with point estimates and standard errors correlating at 0.999 with the main results, with an average absolute difference of 0.003 between the point estimates. All patterns of statistical significance at the 0.05 level are unchanged.

\section{Survey Instrument}\label{app:survey_instrument}

An anonymized copy of the full survey instrument can be found \chref{\urlproject}{here} on OSF. The $Y(T)$ and $Y(T,M)$-experiments are standard; the demographic questions are also standard and closely follow either \cite{kirkland_candidate_2018} or standard ways of eliciting this information from American survey respondents. As it is less common, we discuss the $M(T)$-experiment in more detail. Before proceeding to the task, we tell the respondents

\begin{quote}
We are now going to present you with a series of five potential candidates for mayor. Imagine that each person was running for mayor in your city or town and that the only information available to you is what we present. You will be asked to make your best guess as whether each candidate is an Independent, a Republican or a Democrat.

There are no incorrect answers; please just answer to the best of your ability.
\end{quote}

Then, we show them a hypothetical candidate in the standard conjoint format with instructions that say ``If you had to make your best guess, do you think the candidate [above] is most likely to be an Independent, a Republican, or a Democrat?'' The above link shows how it would appear in the survey. Each respondent performs the task five times. For each respondent, we randomize the order in which the party options (Democrat, Republican, and Independent) appear to avoid order effects. We do this for both the instructions and the questions.

\end{document}

%% file: figures/tab_quota.tex
\multirow{4}{*}{\shortstack{Quota\\Party ID}} &  Democrat & 1325 & 9 & 13 & 29.6\% \\
& Independent & 803 & 688 & 438 & 42.4\% \\
& Republican & 21 & 8 & 1244 & 28\% \\
& Margin & 47.2\% & 15.5\% & 37.3\% &  